\documentclass{article}
\usepackage{ijcai22}

\usepackage{times}

\usepackage{soul}
\usepackage{url}
\usepackage[hidelinks]{hyperref}
\usepackage[small]{caption}
\usepackage{graphicx}
\usepackage{amsmath}
\usepackage{booktabs}
\usepackage{latexsym}

\usepackage{amsthm}
\usepackage{amssymb}
\usepackage{apacite}
\usepackage{natbib}
\usepackage[table]{xcolor}
\usepackage{graphicx}
\usepackage[english]{babel}
\usepackage{pifont}
\newcommand{\xmark}{\ding{55}}%
\newtheorem{proposition}{Proposition}
\newtheorem{theorem}{Theorem}
\newtheorem{lemma}{Lemma}
\newtheorem{corollary}{Corollary}
\newtheorem{definition}{Definition}
\newtheorem{example}{Example}
\newtheorem{remark}{Remark}
\usepackage{enumerate}
\usepackage{tikz}
\usepackage{subcaption}
\usepackage{adjustbox}
\usepackage{colortbl}
\newsavebox{\mytable}
\renewcommand{\paragraph}[1]{{\vspace{\baselineskip}\noindent\normalfont\bfseries#1\quad}}

\title{
Proportional Budget Allocations: Towards a Systematization
}

\author{
Maaike Los$^1$
\and
Zo\'{e} Christoff$^1$\and
Davide Grossi$^{1,2}$
\affiliations
$^1$University of Groningen\\
$^2$University of Amsterdam
\emails
\{m.d.los, z.l.christoff, d.grossi\}@rug.nl
}

\begin{document}

\maketitle

\begin{abstract}
We contribute to the programme of lifting proportionality axioms from the multi-winner voting setting to participatory budgeting. We define novel proportionality axioms for participatory budgeting and test them on known proportionality-driven rules such as Phragm\'{e}n and Rule X. We investigate logical implications among old and new axioms and provide a systematic overview of proportionality criteria in participatory budgeting. 
\end{abstract}

\section{Introduction}
First introduced in Porto Alegre, Brazil in 1988, Participatory budgeting (PB) is a democratic budgeting practice in which citizens are consulted, through some voting method, on how to best allocate a given budget to public projects. The practice is attracting increasing attention from both democracy practitioners worldwide and researchers, among others within computational social choice \citep{aziz_shah_pb}. 

The axiomatic study of PB has formulated a number of criteria for the desirable behaviour of participatory budgeting methods, or PB rules. Special attention has been dedicated to forms of `fairness' or `proportionality'. 
Intuitively, one may want a PB rule to output a division of the available budget over the projects that `reflects' divisions in the voters' preferences. 
A variety of proportionality axioms has been proposed in recent literature, and  this paper provides a first systematization of the axiomatic landscape of proportionality in PB and its special case of multi-winner voting (MWV, \citep{faliszewski17multiwinner}),
or committee selection, that is, a PB setting where all projects (referred to as candidates) have identical cost. 


\paragraph{State of the art.}
The current understanding of proportionality in PB is rooted in MWV \citep{skowron17proportional,peters18prop}.
A key fairness axiom in MWV is {\em justified representation} (JR) \citep{Aziz2017}. In short, JR requires that if a large enough group of voters agrees about a candidate, there is at least one candidate in the chosen committee that at least one of the group members approves. Proportionality requirements have been added with the axioms of {\em extended justified representation} (EJR) \citep{Aziz2017} and {\em proportional justified representation} (PJR) \citep{sanchezfernandez2016}, following the intuition that if a larger group agrees about more candidates, they should be represented by more candidates in the winning committee.
\cite{Aziz2018} generalize these concepts from MWV to PB within an approval voting framework.
A related fairness concept is the \textit{core} \citep{Fain2016}. 
A set of projects, or {\em bundle}, is a core bundle if there is no subset of agents who can afford a different bundle (with their own share of the total budget) where every agent in that subset gets more utility than in the chosen bundle.

In \citep{Peters2020_lim_of_welf}, two MWV rules considered to be proportional---Proportional Approval Voting (PAV) and Phragm\'{e}n's rule (or simply Phragm\'{e}n)---are analysed, and shown to guarantee different types of proportionality. PAV induces a fair distribution of \emph{welfare}, so every group of agents gets a utility proportional to its size, 
while Phragm\'{e}n can be seen as inducing a fair distribution of \emph{power}: the influence of a group of agents is proportional to its size. 
The authors introduce two new proportionality axioms: \textit{priceability} and \textit{laminar proportionality} (LP), and a new rule, Rule X, that is similar to both PAV and Phragm\'{e}n, but satisfies both new axioms and the aforementioned EJR.
Finally, work has started exploring the logical relations between the proportionality axioms proposed in the literature. For instance, \cite{petersmarket} show that, in MWV, the core implies EJR, PJR, and JR. 
Moving to PB, research has focused on assessing the extent to which the above MWV axioms and rules can be meaningfully generalised to PB. \cite{peters2020proportionalPB} generalise Rule X and EJR to PB, and show that, even in this context, Rule X satisfies EJR.\footnote{\cite{peters2020proportionalPB} call Rule X in the PB setting `method of equal shares' (MES). For simplicity we stick to the term Rule X.}
The PB variant of Rule X is also shown to satisfy an approximation of the core and the axiom of priceability for PB. PAV is generalised to PB too, and shown to fail EJR without the unit-cost assumption.

\paragraph{Contribution.}
We make two contributions. First, we complete the study of proportionality for Phragm\'{e}n and Rule X in the PB setting with respect to the axioms mentioned above (Table \ref{tab:rules_and_axioms_PB}). To do so, we propose novel generalizations of PJR and LP (Definitions \ref{def:PB-PJR} and \ref{def:lamprop_PB}) from the MWV to the PB setting. Second, we provide an overview of the logical relations between proportionality axioms in MWV and PB, establishing several novel results (see Figure \ref{fig:relations}). 
The resulting picture contributes to a systematization of how proportionality is interpreted in PB. Only selected proofs and proof sketches are provided in the main text, this version extends the IJCAI 2022 version with an appendix containing full proofs.

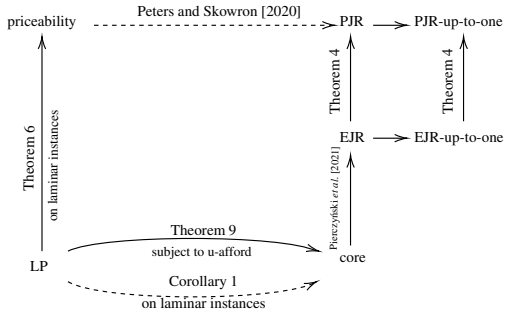
\begin{figure}[t]
    \centering
    \scalebox{0.6}{
    \tikzset{every picture/.style={line width=0.75pt}} 

\scalebox{0.95}{
\begin{tikzpicture}[x=0.75pt,y=0.75pt,yscale=-1,xscale=1]
\draw (409,380) node [anchor=north west][inner sep=0.75pt]   [align=left] {PJR};
\draw (409,483) node [anchor=north west][inner sep=0.75pt]   [align=left] {EJR};
\draw (115,380) node [anchor=north west][inner sep=0.75pt]   [align=left] {priceability};
\draw (409,592) node [anchor=north west][inner sep=0.75pt]   [align=left] {core};
\draw (135,597) node [anchor=north west][inner sep=0.75pt]  [color={rgb, 255:red, 0; green, 0; blue, 0 }  ,opacity=1 ] [align=left] {LP};
\draw (230,371) node [anchor=north west][inner sep=0.75pt]   [align=left] {\cite{Peters2020_lim_of_welf}};
\draw (150.6,565) node [anchor=north west][inner sep=0.75pt]  [opacity=1 ,rotate=-270.64] [align=left] {{\footnotesize \textcolor[rgb]{0,0,0}{on laminar instances}}};
\draw (128.83,538.3) node [anchor=north west][inner sep=0.75pt]  [opacity=1 ,rotate=-270.99] [align=left] {Theorem \ref{thm:LP-pr_PB}};
\draw (400,590) node [anchor=north west][inner sep=0.75pt]  [opacity=1, rotate=-270.99 ] [align=left] {\scriptsize \cite{peters2020proportionalPB}}; 
\draw (400,470) node [anchor=north west][inner sep=0.75pt]  [opacity=1 ,rotate=-269.61] [align=left] {Theorem \ref{thm:PB-EJR->PB-PJR}};
\draw (290,565) node [anchor=north] [inner sep=0.75pt]   [align=left] {Theorem \ref{thm:LP-core_u-afford}};
\draw (290,585) node [anchor=north] [inner sep=0.75pt]   [align=left] {{\footnotesize subject to 
u-afford
}};
\draw [draw opacity=1]   (170,590) .. controls (200,575) and (360,575) .. (390,590) ;
\draw [shift={(392,590)}, rotate = 555] [draw opacity=1 ][line width=0.75]    (10.93,-3.29) .. controls (6.95,-1.4) and (3.31,-0.3) .. (0,0) .. controls (3.31,0.3) and (6.95,1.4) .. (10.93,3.29)   ;
\draw [draw opacity=1, dashed ]   (170,617) .. controls (200,635) and (360,635) .. (390,617) ;
\draw [shift={(390,617)}, rotate = 515] [draw opacity=1 ][line width=0.75]    (10.93,-3.29) .. controls (6.95,-1.4) and (3.31,-0.3) .. (0,0) .. controls (3.31,0.3) and (6.95,1.4) .. (10.93,3.29)   ;
\draw (290,610) node [anchor=north] [inner sep=0.75pt]   [align=left] {Corollary \ref{cor:LP-PJR-EJR-core_MWV}};
\draw (290,630) node [anchor=north] [inner sep=0.75pt]   [align=left] {on laminar instances};
\draw [draw opacity=1 ]   (147,590) -- (147,400) ;
\draw [shift={(147,400)}, rotate = 450.28] [draw opacity=1 ][line width=0.75]    (10.93,-3.29) .. controls (6.95,-1.4) and (3.31,-0.3) .. (0,0) .. controls (3.31,0.3) and (6.95,1.4) .. (10.93,3.29)   ;
\draw [draw opacity=1 , dashed]   (193,390) -- (404,390) ;
\draw [shift={(406,390)}, rotate = 180.15] [draw opacity=1 ][line width=0.75]    (10.93,-3.29) .. controls (6.95,-1.4) and (3.31,-0.3) .. (0,0) .. controls (3.31,0.3) and (6.95,1.4) .. (10.93,3.29)   ;
\draw [draw opacity=1]   (420,475) -- (420,400) ;
\draw [shift={(420,400)}, rotate = 449.25] [draw opacity=1 ][line width=0.75]    (10.93,-3.29) .. controls (6.95,-1.4) and (3.31,-0.3) .. (0,0) .. controls (3.31,0.3) and (6.95,1.4) .. (10.93,3.29)   ;
\draw [draw opacity=1]   (420,585) -- (420,505) ;
\draw [shift={(420,505)}, rotate = 449.81] [draw opacity=1 ][line width=0.75]    (10.93,-3.29) .. controls (6.95,-1.4) and (3.31,-0.3) .. (0,0) .. controls (3.31,0.3) and (6.95,1.4) .. (10.93,3.29)   ;
\draw [draw opacity=1]   (440,490) -- (470,490) ;
\draw [shift={(470,490)}, rotate = 180.15] [draw opacity=1 ][line width=0.75]    (10.93,-3.29) .. controls (6.95,-1.4) and (3.31,-0.3) .. (0,0) .. controls (3.31,0.3) and (6.95,1.4) .. (10.93,3.29)   ;
\draw [draw opacity=1]   (440,390) -- (470,390) ;
\draw [shift={(470,390)}, rotate = 180.15] [draw opacity=1 ][line width=0.75]    (10.93,-3.29) .. controls (6.95,-1.4) and (3.31,-0.3) .. (0,0) .. controls (3.31,0.3) and (6.95,1.4) .. (10.93,3.29)   ;
\draw [draw opacity=1]   (520,475) -- (520,400) ;
\draw [shift={(520,400)}, rotate = 449.25] [draw opacity=1 ][line width=0.75]    (10.93,-3.29) .. controls (6.95,-1.4) and (3.31,-0.3) .. (0,0) .. controls (3.31,0.3) and (6.95,1.4) .. (10.93,3.29)   ;
\draw (475,380) node [anchor=north west][inner sep=0.75pt]   [align=left] {PJR-up-to-one};
\draw (475,483) node [anchor=north west][inner sep=0.75pt]   [align=left] {EJR-up-to-one};
\draw (500,470) node [anchor=north west][inner sep=0.75pt]  [opacity=1 ,rotate=-269.61] [align=left] {Theorem \ref{thm:PB-EJR->PB-PJR}};

\end{tikzpicture}
}
    }
    \caption{Relations among proportionality axioms in PB. Dashed lines indicate relations that only hold in MWV. Arrows are labelled either by our results or the paper where they have been proven.
    Some of the implications only hold under certain conditions or restrictions: \textit{laminar instances} (Definition \ref{def:laminar_PB}), and \textit{unanimity affordability} (u-afford, Definition \ref{def:P-u-afford}). 
    Transitive arrows are omitted. Absence of arrows denotes the existence of a counterexample. 
    }
    \label{fig:relations}
\end{figure}


\section{Preliminaries}
\subsection{The Participatory Budgeting (PB) Problem}
We denote the set of projects (or candidates) by $C=\{c_1, c_2, ..., c_m\}$ and the set of voters by $N=\{v_1,v_2,...,v_n\}$. Each voter $i$ comes with a function $u_i$ assigning a utility to all projects. 
In the fully general setting of utility-based PB, $u_i: C\rightarrow [0,1]$. In contrast, in the special case of approval-based PB, the utility function is restricted to two values: $u_i: C\rightarrow \{0,1\}$, determining voter $i$'s approval set $A_i=\{c\in C:u_i(c)=1\}$. 
The utility of a set of projects $T \subseteq C$ for a set of voters $S \subseteq N$ is defined additively as: $u_S(T)=\sum_{i\in S}\sum_{c\in T}u_i(c)$. 
A \textit{profile} $P$ is a vector of the utility functions of all voters: $P=(u_1, ..., u_n)$. If utilities are approval-based we denote with $C(P)$ the set of all projects occurring in the approval sets in $P$.
A function $\text{cost: } C\rightarrow\mathbb{Q}_+$ 
assigns a cost to every project. The cost of a set of projects $T$ is given by $\text{cost}(T)=\sum_{c\in T}\text{cost}(c)$. The total budget is denoted by $l$. If $l$ is not mentioned, it is equal to $1$. Hence, an election instance (also called a \textit{PB-instance}) $E=(N,C,\text{cost}, P, l)$ consists of a set of voters $N$, a set of projects $C$, a cost function, a profile $P$, and a budget $l$. If all else is clear in context, we abbreviate this to $E=(P,l)$. A \textit{voting rule} $\mathcal{R}$ maps an election instance $E$ to a winning bundle $W$. A PB-instance where for all $i\in N$ $u_i: C\rightarrow \{0,1\}$ is called \textit{approval-PB-instance}.
The special case of an approval-PB-instance in which all projects have the same cost is called a \textit{MWV-instance}. In such MWV setting, we refer to a bundle as a \textit{committee}.

\subsection{Voting Rules}\label{sec:rules}

We focus on three proportionality-inspired rules: (sequential) Phragm\'{e}n, proportional approval voting (PAV) and Rule X. Below, we recall the generalizations of Phragm\'{e}n, PAV, and Rule X to PB 
introduced by \cite{peters2020proportionalPB}. 

\smallskip
\noindent
{\bf Phragm\'{e}n}
Every voter gets currency continuously at the rate of one unit of currency per unit of time. At the first moment $t$ when there is a group of  voters $S$ who all approve a not-yet-selected project $c$, and who together have cost$(c)$ units of currency, the rule adds $c$ to the bundle and asks the voters from $S$ to pay the cost of $c$ (i.e., the rule resets the balance of each voter from $S$), while the others keep their so-far earned money. The process stops when it would select a project which would overshoot the budget.

\smallskip
\noindent
{\bf PAV} 
The winning bundle $W$ of PAV is the bundle with $\text{cost}(W)\leq l$ that maximises the score $\text{PAV-score}(W)  = \sum_{i\in N}\left(1 +\frac{1}{2} + \frac{1}{3}+\cdots+\frac{1}{|W\cap A_i|} \right)$.

\smallskip
\noindent
{\bf Rule X} The rule starts by giving each voter an equal fraction of the budget. In case of a budget of 1, each of the $n$ voters gets $\frac{1}{n}$ unit of currency. We start with an empty bundle $W = \emptyset$ and sequentially add projects to $W$. To add a project $c$ to $W$, the voters have to pay for $c$. Write $p_i(c)$ for the amount that voter $i$ pays for $c$; we will need that $\sum_{i\in N}p_i(c) = \text{cost}(c).$ Let $p_i(W) = \sum_{c\in W}p_i(c) \leq \frac{1}{n}$ be the total amount voter $i$ has paid so far. For $\rho \geq 0$, we say that a project $c\notin W$ is $\rho$-affordable if $\sum_{i\in N}\min(\frac{1}{n}-p_i(W), u_i(c)\cdot \rho) = \text{cost}(c).$
The rule iteratively selects a project $c\notin W$ that is $\rho$-affordable for a minimum $\rho$. Individual payments are given by $p_i(c) = \min(\frac{1}{n}-p_i(W), u_i(c)\cdot \rho).$
If no project is $\rho$-affordable for any $\rho$, Rule X terminates and returns $W$. 

\begin{remark}
Note that Phragm\'{e}n and PAV work with approval-based while Rule X with utility-based ballots. In what follows, when discussing the first two rules, we will therefore presuppose approval-PB-instances.
Notice also that all three rules are non-resolute.
\end{remark}

\begin{example}\label{ex:rules}
Consider the profile in Table \ref{tab:ex_1}. 
To get the approval profile needed for Phragm\'{e}n and PAV, we binarize the utility function using a threshold
of 0.3: voters approve a project when it yields utility of at least 0.3. Approved projects are shaded.
    \begin{table}[t]
    \begin{adjustbox}{width=\columnwidth,center=\columnwidth}
        \centering
        \begin{tabular}{|c|l||r|r|r|r||c|c|c|}
        \hline
             project & cost & \multicolumn{4}{c||}{utilities} & \multicolumn{3}{c|}{Winning bundles} \\ \hline
              $c_1$ & 0.4  & \cellcolor[HTML]{C0C0C0}{1} & \cellcolor[HTML]{C0C0C0} {0.7} & 0.1 & 0 & & & \checkmark \\ \hline
            $c_2$ & 0.3  & \cellcolor[HTML]{C0C0C0} {0.3} & \cellcolor[HTML]{C0C0C0} {0.4} & 0 &\cellcolor[HTML]{C0C0C0}  {0.4} & $\checkmark$ & \checkmark &\\ \hline
             $c_3$ & 0.7 & 0.1 & 0.2 & \cellcolor[HTML]{C0C0C0} {0.4} & \cellcolor[HTML]{C0C0C0} {0.4}& & \checkmark & \\ \hline
             $c_4$ & 0.35& 0 &\cellcolor[HTML]{C0C0C0}  {0.4} & 0.2 &\cellcolor[HTML]{C0C0C0}  {1} & \checkmark &  & \checkmark\\ \hline \hline
                   & & $v_1$ & $v_2$ & $v_3$ & $v_4$ & Phragm\'{e}n & PAV & Rule X \\
        \hline
        \end{tabular}    
    \end{adjustbox}
    \caption{The profile used in Examples \ref{ex:rules} and \ref{ex:PB-PJR}. Each column contains the utilities per project of a voter. 
    The budget $l=1$.}
    \label{tab:ex_1}

    \end{table}
Phragm\'{e}n will first select $c_2$ at time $t=0.1$, which leaves voter $v_3$ with 0.1 units of currency. Then at $t=0.275$, $v_2$ and $v_4$ can together buy $c_4$, which leaves $v_1$ with 0.175 and $v_3$ with 0.275. After adding $c_4$, the rule ends since both remaining projects are not affordable and outputs $W=\{c_2, c_4\}$. For PAV, we compute the PAV-score of four sets: $\text{PAV-score}(\{c_1, c_4\})=3.5$, $\text{PAV-score}(\{c_1, c_2\})=4 $, $\text{PAV-score}(\{c_2, c_3\})=4.5 $, and $\text{PAV-score}(\{c_2, c_4\})=4$ (clearly, smaller sets have a lower score and larger sets are not affordable). Hence, PAV will select $W=\{c_2, c_3\}$.
Rule X starts by giving every voter $\frac{1}{4}$ unit of currency. First, $c_4$ is $\rho$-affordable for $\rho = \frac{0.35}{1.6}\approx 0.219$, then, for $\rho = \frac{0.4}{1.8}\approx0.222$, $c_1$ is $\rho$-affordable. After selecting $c_1$, the sum of the remaining amounts of all voters is 0.25, so the other projects are not $\rho$-affordable for any $\rho$. Hence, Rule X returns $W=\{c_1, c_4\}$.
\end{example}

\subsection{Known Proportionality Axioms}

We recall axioms for PB from \cite{peters2020proportionalPB} that generalize known MWV axioms. 
We start with the axioms of core and extended justified representation (EJR).

\begin{definition}[Core] \label{def:Core_PB} For a given PB-instance $E=(N,C,\text{cost}, P, l)$, a  bundle $W$ is in the core if for every $S \subseteq N$ and $T \subseteq C$ with $|S|\geq\frac{\text{cost}(T)}{l}\cdot n$ there exists $i \in S$ such that $u_i(W) \geq u_i(T)$. 
A voting rule $\mathcal{R}$ satisfies the core property if for each PB-instance $E$ the winning bundle $\mathcal{R}(E)$ is in the core. 
\end{definition}

To define EJR, we first introduce \textit{($\alpha,T$)-cohesiveness}:
\begin{definition}[($\alpha,T$)-cohesiveness]\label{def:alpha-T-cohesiveness}
 A group of voters S is \textit{($\alpha,T$)-cohesive} for $\alpha: C\rightarrow[0,1]$ and $T\subseteq C$, if $|S|\geq \frac{\text{cost}(T)}{l}\cdot n$ and if it holds that $u_i(c)\geq \alpha (c)$ for every voter $i\in S$ and each project $c\in T$.
\end{definition}

\begin{definition}[Extended justified representation]\label{def:EJR_PB} 
    A rule $\mathcal{R}$ satisfies extended justified representation (EJR) if for each PB-instance $E$ and each $(\alpha,T)$-cohesive group of voters $S$, there is a voter $i\in S$ such that $u_i(\mathcal{R}(E))\geq\sum_{c\in T}\alpha(c)$.
\end{definition}
In \cite{peters2020proportionalPB}, a weakening of EJR is considered, to which we refer as \textit{EJR-up-to-one}.
\begin{definition}[Extended justified representation up to one project]\label{def:EJR_PB-up-to-one} 
    A rule $\mathcal{R}$ satisfies extended justified representation up to one project (EJR-up-to-one)
    if for each PB-instance $E$ and each $(\alpha,T)$-cohesive group of voters $S$, there is a voter $i\in S$ such that $u_i(\mathcal{R}(E))\geq\sum_{c\in T}\alpha(c)$ or for some $a\in C$ it holds that $u_i(\mathcal{R}(E)\cup \{a\})>\sum_{c\in T}\alpha(c)$. 
\end{definition}

In MWV, Definitions~\ref{def:EJR_PB} and~\ref{def:EJR_PB-up-to-one} are equivalent:
\begin{proposition}\label{prop:ejr-up-to-one}
Def.~\ref{def:EJR_PB} and Def.~\ref{def:EJR_PB-up-to-one} are equivalent in MWV-instances.
\end{proposition}

We now turn to priceability, for which the notion of price system needs to be introduced first.
\begin{definition}[Price systems] \label{def:Pricesystems_PB}
     A price system is a pair $\mathbf{ps}=(b,(p_i)_{i\in N})$ where $b\geq 1$ is the initial budget, and for each voter $i\in N$, there is a payment function $p_i:C\rightarrow \mathbb{R}$ such that (1) a voter can only pay for projects she gets at least some utility from: if $u_i(c)=0$, then $p_i(c)=0$ for each $i\in N$ and $c\in C$, and (2) each voter can spend the same budget of $\frac{b}{n}$ units of currency: $\sum_{c\in C}p_i(c)\leq \frac{b}{n}$ for each $i\in N$.
\end{definition}
    
\begin{definition}[Priceability] \label{def:Priceability_PB} 
    A rule $\mathcal{R}$ satisfies \textit{priceability} (is \textit{priceable}) if for each PB-instance $E$, there exists a price system $\mathbf{ps}=(b,(p_i)_{i\in N})$ that supports $\mathcal{R}(E)$, that is: (1) for each $c \in \mathcal{R}(E) $, the sum of the payments for $c$ equals its price, i.e., $\sum_{i\in N}p_i(c)=\text{cost}(c)$; (2) no project outside of the winning bundle gets any payment, i.e., for all $c\notin \mathcal{R}(E), \sum_{i\in N}p_i(c)=0$; (3) there exists no non-selected project whose supporters in total have a remaining unspent budget of more than its cost, i.e., for all $c\notin \mathcal{R}(E)$, 
            $
            \sum_{i\in N \text{ s.t. } u_i(c)>0}\left(\frac{b}{n}-\sum_{c'\in \mathcal{R}(E)}p_i(c')\right)\leq \text{cost}(c).
            $
\end{definition}

\begin{remark}
The above properties are defined for PB. Throughout the paper, when needing to refer to the MWV specialization of a PB axiom, we will assume the axiom is defined on MWV-instances. We say that a rule satisfies an axiom in \textit{MWV-instances} if, for all MWV-instances, the rule's winning bundle satisfies the axiom.
Observe also that, since Phragm\'{e}n and PAV are defined on approval-PB-instances, when we assess whether they satisfy an axiom, we consider the axiom only with respect to approval-PB-instances.
\end{remark}

\section{Two Novel Axioms for PB}
\subsection{Proportional Justified Representation in PB}
\cite{sanchezfernandez2016} define proportional justified representation for MWV (we refer to it as MWV-PJR here).
We generalize this axiom to PB based on the generalisation of EJR provided by \cite{peters2020proportionalPB}. Two steps are involved: dropping the unit-cost assumption, and allowing arbitrary utilities instead of just approval ones.

\begin{definition}[Proportional justified representation]\label{def:PB-PJR}
    A rule $\mathcal{R}$ satisfies \textit{proportional justified representation} (PJR) if for each PB-instance $E$ and $(\alpha,T)$-cohesive group of voters $S$, 
    \begin{equation}\label{eq:PB-PJR}
        \sum_{c\in \mathcal{R}(E)}(\max_{i\in S}u_i(c))\geq\sum_{c\in T}\alpha(c).
    \end{equation}
\end{definition}
 The intuition is that, in the winning bundle, for each cohesive group $S$ (cohesive in that $S$ agrees to a certain degree about the set of projects $T$) there should be enough projects to which at least one voter in $S$ assigns enough utility.

\begin{example}\label{ex:PB-PJR}
Consider the profile in Table \ref{tab:ex_1}, with the bundle $W=\{c_2,c_3\}$ as selected by PAV. Now consider the group $S=\{v_1, v_2\}$. Probably both voters in $S$ are happy that $c_2$ is selected, but they would both get more utility from $c_1$ than from the selected $c_2$ or $c_3$. Also, if each of them would get their share of the total budget ($\frac{1}{4}$), they could together afford the set $T = \{c_1\}$. Intuitively then, $W$ is not a fair bundle considering voters $v_1$ and $v_2$. Let us look at it more formally. 
$S$ is $(\alpha, T)$-cohesive for $\alpha(c_1)=0.7$ ($\alpha(c_2)$, $\alpha(c_3)$ and $\alpha(c_4)$ are arbitrary) since the voters in $S$ can afford $T$ with their share of the budget, and for both of them $u(c_1)\geq \alpha(c_1)$. However, Equation \ref{eq:PB-PJR} is not satisfied: $\sum_{c\in W}(\max_{i\in S}u_i(c)) = 0.4+0.2=0.6$, while $\sum_{c\in T}\alpha(c)=0.7.$ This shows that in the given election instance, the bundle $W$ does not satisfy PJR.
\end{example}

Our definition of PJR for PB is rather different from its MWV variant. It requires some work to show that the proposed definition reduces to the definition of MWV-PJR under unit-cost assumption and approval preferences. First of all, let us recall the definition of MWV-PJR:
\begin{definition}[PJR for MWV \citep{sanchezfernandez2016}] \label{def:PJR_MWV}
 An approval based voting rule $\mathcal{R}$ satisfies PJR for MWV (MWV-PJR) if for every ballot profile $P$ and committee size $k$, the rule outputs a committee $W = \mathcal{R}(P,k)$ s.t.:
for every $\ell\leq k$ and every $\ell$-cohesive set of voters $S\subseteq N$, it holds that $|W\cap\left(\cup_{i\in S}A_i\right)|\geq \ell$, where a set $S$ is $\ell$-cohesive if $|S|\geq \ell\cdot \frac{n}{k}$ and $|\cap_{i\in S}A_i|\geq \ell$.
\end{definition}

We will need the following lemma:

\begin{lemma} \label{lemma_1}
Let $E=(N,C,\text{cost}, P, l)$ be an approval-PB-instance where for all projects $c$, $\text{cost}(c)=\frac{1}{k}$.
Then:
    (a) for given $\alpha:C\rightarrow[0,1] $ and $T\subseteq C$,  any group of voters $S\subseteq N$ that is $(\alpha, T)$-cohesive is also $\ell$-cohesive for $\ell=|T'|$ with $T'= \{c\in T: \alpha(c)> 0\}$; and
    (b) for every group $S$ that is $\ell$-cohesive there are $T\subseteq C$ with $|T|=\ell$ and $\alpha: C\rightarrow [0,1]$ with $\alpha(c)=1$ for all $c\in C$, such that $S$ is $(\alpha,T)$-cohesive.
\end{lemma}

\begin{theorem}\label{thm:PB-PJR}
MWV-PJR and PJR are equivalent in MWV-instances.
\end{theorem}
\begin{proof}[Proof Sketch]
We show that on MWV-instances, Definitions \ref{def:PB-PJR} and \ref{def:PJR_MWV} are equivalent.
\fbox{PJR $\Rightarrow$MWV-PJR} Assume that a rule $\mathcal{R}$ satisfies PJR, and take an arbitrary MWV-instance $E$. 
Because $E$ satisfies the assumptions of unit-cost and approval based voting, the fact that $\mathcal{R}$ satisfies PJR boils down to the following: for all $S, \alpha: C\rightarrow[0,1]$, and  $T\subseteq C$ with $|S| \geq |T|\cdot \frac{n}{k}$ and for which $u_i(c)\geq \alpha(c) $ for all $ i\in S$ and for all $c \in T$, it is the case that $|\mathcal{R}(E)\cap(\cup_{i\in S} A_i)|\geq\sum_{c\in T}\alpha(c).$
Take arbitrary $S\subseteq N$ and $\ell\leq k$ and suppose that $S$ is $\ell$-cohesive. According to Lemma \ref{lemma_1}(b), there are $T\subseteq C$ with $|T|=\ell$ and $\alpha: C\rightarrow [0,1]$ with $\alpha(c)=1$ for all $c\in C$, such that $S$ is $(\alpha,T)$-cohesive. Because $\mathcal{R}$ satisfies PJR, this implies that $|\mathcal{R}(E)\cap(\cup_{i\in S} A_i)|\geq\sum_{c\in T}\alpha(c)$. However, because of our choice of $T$ and $\alpha$, we know that $\sum_{c\in T}\alpha(c)=|T|=\ell$, so $|\mathcal{R}(E)\cap(\cup_{i\in S} A_i)|\geq \ell$, which shows that $\mathcal{R}$ satisfies MWV-PJR.
\fbox{MWV-PJR $\Rightarrow$PJR}
Assume that a rule $\mathcal{R}$ satisfies MWV-PJR. Take arbitrary MWV-instance $E$, and suppose that a group $S$ is $(\alpha,T)$-cohesive. Then according to Lemma \ref{lemma_1}(a), when we take $T'= \{c\in T: \alpha(c)> 0\}$, $S$ is $\ell$-cohesive for $\ell=|T'|$. 
Because $R$ satisfies MWV-PJR, it follows that $|\mathcal{R}(E)\cap(\cup_{i\in S} A_i)|\geq \ell = |T'|$. By definition of $T'$, $|T'|\geq \sum_{c\in T'}\alpha(c) = \sum_{c\in T}\alpha(c)$, so $|\mathcal{R}(E)\cap(\cup_{i\in S} A_i)|\geq\sum_{c\in T}\alpha(c)$ as desired.
\end{proof}

\begin{remark}
To our knowledge, besides Definition \ref{def:PB-PJR}, the only generalization of MWV-PJR to PB is from \cite{Aziz2018}. 
They define an axiom called \textit{Strong-BPJR-L} (where L stands for the budget limit, to which we refer here as $\ell$)
that requires the following: For a budget $l$, a bundle  $W$ satisfies Strong-BPJR-L if for all $\ell \in [1, l]$ there does not exist a set of voters $S\subseteq N$ with $|S| \geq \ell \frac{n}{l}$, such that cost$(\cap_{i\in S}A_i) \geq \ell$ but cost$((\cup_{i\in S} A_i)\cap W)< \ell$. It is possible to generalise this definition further to allow arbitrary utilities instead of approval votes. However, note that the requirement in this definition is not that for every $\ell$-cohesive $S$ the \textit{utility} of the projects they all approve that are selected is at least $\ell$, but rather the \textit{cost} of this set of projects. Although this is indeed a generalisation of MWV-PJR, as is shown by \cite{Aziz2018}, we consider that the aim of PJR is to ensure a certain level of utility for every group of voters, rather than a certain cost. Definition \ref{def:PB-PJR} is equivalent to Strong-BPJR-L when assuming that a project's cost is directly proportional to a voter's utility from it.
\end{remark}

Like for EJR (Definition \ref{def:EJR_PB-up-to-one}), we can add an up-to-one-project condition to PJR: 

\begin{definition}[Proportional justified representation up to one project]\label{def:PB-PJR-up-to-one}
    A rule $\mathcal{R}$ satisfies \textit{proportional justified representation up to one project} (PJR-up-to-one) if for each PB-instance $E$ and each $(\alpha,T)$-cohesive group of voters $S$, 
    $\sum_{c\in \mathcal{R}(E)}(\max_{i\in S}u_i(c))\geq\sum_{c\in T}\alpha(c)$ or for some  $a \in C$  it holds that $\sum_{c\in \mathcal{R}(E)\cup\{a\}}(\max_{i\in S}u_i(c))>\sum_{c\in T}\alpha(c).$
\end{definition}

\subsection{Laminar Proportionality in PB} \label{sec:LP-in-PB}
The basic idea of \textit{LP} for MWV \citep{Peters2020_lim_of_welf} is that if we know about a strict separation between different parties, we can divide the chosen projects proportionally over the parties. We generalize this notion to PB in the approval voting setting by taking the budget $l$ instead of the bundle size $k$, and by using the cost of each project instead of unit-cost. 

\begin{definition}[Laminar PB-instances]\label{def:laminar_PB}
An approval-PB-instance $(P,l)$ is \textit{laminar} if either: (1) $P$ is unanimous and $\text{cost}(C(P))\geq l$; (2) there is $c\in C(P)$ such that $c\in A_i$ for all $A_i\in P$, the profile $P_{-c}$ (i.e., $P$ once we remove $c$) is not unanimous and instance $(P_{-c}, l-\text{cost}(c))$ is laminar (with $P_{-c}=(A_1\backslash \{c\}, ..., A_n\backslash\{c\})$); or (3) There are two laminar PB-instances $(P_1, l_1)$ and $(P_2, l_2)$ with $C(P_1)\cap C(P_2)=\emptyset$ and  $|P_1|\cdot l_2 = |P_2|\cdot l_1$ such that $P=P_1+P_2$ and $l=l_1+l_2$.
\end{definition}

\begin{example}\label{ex:laminar_PB}
\begin{table}[t]
    \centering
        \setlength{\tabcolsep}{20pt}
    \begin{adjustbox}{width=0.6\columnwidth, center=\columnwidth}
    \begin{tabular}{c c c }
    \cline{1-2}
        \multicolumn{2}{|c|}{$c_3$, 3} &   \\ \hline
        \multicolumn{2}{|c|}{\cellcolor[HTML]{C1C1C1}$c_2$, 3} &   \multicolumn{1}{c|}{$c_5$, 4}  \\ \hline
        \multicolumn{2}{|c|}{\cellcolor[HTML]{C1C1C1}$c_1$, 2} &   \multicolumn{1}{c|}{\cellcolor[HTML]{C1C1C1}$c_4$, 2} \\ \hline
        \multicolumn{3}{|c|}{\cellcolor[HTML]{C1C1C1}$c_6$, 1}  \\ \hline
        $v_1$ & $v_2$ & $v_3$ 
    \end{tabular}
    \end{adjustbox}
    \caption{Example of a laminar proportional bundle $W$ (shaded) in a laminar election instance.  Each column represents the approval set of a voter (written beneath it), and each box shows a project with its cost.  E.g., voter $v_2$ approves $c_6,c_1,c_2,$ and $c_3$ and cost$(c_5)=4$. 
    }
    \label{tab:ex_lamprop}
\end{table}
The instance $P$ in Table \ref{tab:ex_lamprop} associated with the budget $l=10$ is laminar. The instance $P_1$ with $v_1$ and $v_2$ and projects $c_1, c_2$, and $c_3$ with limit $l_1 =6$ satisfies the first item of Definition \ref{def:laminar_PB}, as does the instance $P_2$ with only voter $v_3$ and projects $c_4$ and $c_5$, and limit $l_2=3$. Those two instances can be added by Definition \ref{def:laminar_PB}, item 3, since $|P_1|\cdot l_2 =2\cdot3 = 1\cdot 6 = |P_2|\cdot l_1$. Then $c_6$ can be added by Definition \ref{def:laminar_PB}, item 2, to get $P$ with limit $l=6+3+\text{cost}(c_6) = 10$.
\end{example}

\begin{definition}[Laminar proportionality]\label{def:lamprop_PB}
    A rule $\mathcal{R}$ satisfies laminar proportionality (\textit{LP}) if for every laminar PB-instance $E=(P,l)$, $\mathcal{R}(E) = W$ where $W$ is a \textit{laminar proportional bundle}, i.e.: (1) if $P$ is unanimous, then $W \subseteq C(P)$ (if everyone agrees, then part of the projects they agree on is chosen); (2) if there is a unanimously approved project $c$ s.t. $(P_{-c}, l-\text{cost}(c))$ is laminar, then $W = W'\cup \{c\}$ where $W'$ is laminar proportional for $(P_{-c}, l-\text{cost}(c))$; or (3) If $P$ is the sum of laminar PB-instances $(P_1, l_1)$ and $(P_2, l_2)$, then $W = W_1 \cup W_2$ where $W_1$ is laminar proportional for $(P_1, l_1)$ and $W_2$ is laminar proportional for $(P_2, l_2)$.
\end{definition}
It is trivial that in case of unit-cost and budget $k$, these definitions are equivalent to the corresponding MWV definitions.

\begin{example}
Elaborating on Example \ref{ex:laminar_PB}, bundle $W = \{c_1, c_2, c_4, c_6\}$ (grey in Table \ref{tab:ex_lamprop}) is laminar proportional in that instance  with a budget of $l=10$. In $(P_1, l_1)$, $\{c_1, c_2\}$ is laminar proportional, as is $\{c_4\}$ in $(P_2, l_2)$. Hence, $\{c_1, c_2, c_4\}$ is laminar proportional in $(P_1+P_2, l_1+l_2)$, and $\{c_1, c_2, c_4, c_6\}$ is laminar proportional in $(P, l)$.
\end{example}


\section{Proportionality Properties of Rules}

\begin{table*}[t]
\begin{adjustbox}{width=1.1\textwidth,center=\textwidth}
\begin{tabular}{|l|l|l|l|l|l|l|}
\hline
            & \multicolumn{2}{c|}{\textbf{PAV}}
                    & \multicolumn{2}{c|}{\textbf{Phragm\'{e}n}}
                            & \multicolumn{2}{c|}{\textbf{Rule X} }                \\
    \hline
     & MWV & PB & MWV & PB & MWV & PB\\
    \hline
\textbf{core} & 
 \multicolumn{2}{c|}{\xmark\citep{Aziz2017}} 
 & \multicolumn{2}{c|}{\xmark\citep{Brill2017}} & \xmark \citep{Peters2020_lim_of_welf}& \xmark \cite{peters2020proportionalPB}  \\ \hline
\textbf{EJR} & 
\checkmark \citep{Aziz2017}& \xmark  \citep{peters2020proportionalPB} &\multicolumn{2}{c|}{\xmark\citep{Brill2017}} &  \checkmark \citep{Peters2020_lim_of_welf} & \checkmark up-to-one \citep{peters2020proportionalPB}  \\ \hline
\textbf{PJR} & 
\checkmark  \cite{sanchezfernandez2016} & \cellcolor[HTML]{C0C0C0} \xmark (Prop. \ref{prop:PAV-PJR})
& \checkmark  \citep{Brill2017} &\cellcolor[HTML]{C0C0C0} \checkmark (Prop. \ref{prop:Phragmen-PJR})& \checkmark  \citep{Peters2020_lim_of_welf}& \cellcolor[HTML]{C0C0C0} \checkmark up-to-one (Prop. \ref{prop:RuleX-PJR_PB})  \\ \hline
\textbf{p.bility} &
\multicolumn{2}{c|}{\xmark\citep{Peters2020_lim_of_welf}} 
& \checkmark \citep{Peters2020_lim_of_welf}& \cellcolor[HTML]{C0C0C0} \checkmark (Prop. \ref{prop:Phragmen-priceable})& \checkmark \citep{Peters2020_lim_of_welf} & \checkmark  \citep{peters2020proportionalPB}\\ \hline
\textbf{LP} &
\multicolumn{2}{c|}{\xmark\citep{Peters2020_lim_of_welf}} 
& \checkmark \citep{Peters2020_lim_of_welf} &\cellcolor[HTML]{C0C0C0}\xmark (Prop. \ref{prop:Phragmen-negative})& \checkmark \citep{Peters2020_lim_of_welf} &\cellcolor[HTML]{C0C0C0}\xmark (Prop. 
\ref{prop:Phragmen-negative}) \\ \hline 
\end{tabular}

\end{adjustbox}
\caption{Three rules and the properties they satisfy: \checkmark indicates satisfaction, \xmark ~failure. Shaded entries indicate new results, references to the corresponding propositions or literature are included for each entry. Recall that PAV and Phragm\'{e}n are assessed on approval-PB-instances.
}
\label{tab:rules_and_axioms_PB}
\end{table*}
Table \ref{tab:rules_and_axioms_PB} summarizes the findings of this section. 

The literature already provides several results about PAV: it does not satisfy the core, priceability, or LP in MWV (and hence not in PB either). In MWV-instances PAV satisfies PJR, but it does not in PB: 
\begin{proposition}\label{prop:PAV-PJR}
PAV does not satisfy PJR.
\end{proposition}

We turn now to our analysis of Phragm\'{e}n and Rule X.
\begin{proposition}\label{prop:Phragmen-PJR}
Phragm\'{e}n satisfies PJR.
\end{proposition}

\begin{proof}
Assume towards a contradiction that there exist a group of voters $S\subseteq N$, a set of projects $T\subseteq C$, and a function $\alpha:C\rightarrow [0,1]$ such that $S$ is $(\alpha,T)$-cohesive, and for this $S$, $\alpha$, and $T$, the winning bundle $W$ of Phragm\'{e}n does not contain enough projects that voters from $S$ like enough: $\sum_{c\in W}(\max_{i\in S}u_i(c)) < \sum_{c\in T}\alpha(c)$. Note that in approval-PB-instances this boils down to $|W\cap \cup_{i\in S}A_i|<\sum_{c\in T}\alpha(c)$. Because of $(\alpha,T)$-cohesiveness, for every voter $i\in S$ and each project $c\in T$, $u_i(c)\geq \alpha (c)$, 
so either $\alpha(c)=0$ or $c\in A_i$, and therefore 
$\sum_{c\in T}\alpha(c) \leq |T\cap \cap_{i\in S}A_i|$. We 
write $T'$ for $T\cap \cap_{i\in S}A_i$, and $W'$ for $W\cap \cup_{i\in S}A_i$. Hence, $|W'|<|T'|\leq |T|$. Let $t$ be the moment when the rule stops:  a project $c$ is reached that would overshoot the budget. Clearly, cost$(W)+\text{cost}(c)>1$, but cost$(W)\leq 1$. 
Let $x$ be the amount of virtual money earned by all voters so far, so 
\begin{equation}\label{eq:total}
    t\cdot n = x = \text{cost}(W)+\text{cost}(c)+y,
\end{equation}
where $y\geq 0$ is the money that non-supporters of $c$ have earned in the meantime. 
Because $|W'|<|T'|$, there must be some project in $T$ that is not in $W$. The voters in $S$ together have earned $\frac{x}{n}\cdot |S|$,  and because $S$ is $(\alpha,T)$-cohesive, $|S|\geq \text{cost}(T)\cdot n$, so 
\begin{equation}\label{eq:money_S}
   \frac{x}{n}\cdot |S| \geq \text{cost}(T)\cdot n\cdot \frac{x}{n}= \text{cost}(T)\cdot x. 
\end{equation}
 From (\ref{eq:total}) and the fact that cost$(W)+\text{cost}(c)>1$, it follows that $x>1+y$ (and 
 $x>1$).  From (\ref{eq:money_S}), it follows that the voters in $S$ have earned enough together at time $t$ to buy all projects from $T$ (and therefore from $T'$), but have not done so. Hence, either they have also paid for projects not in $T'$, or, if they only spent their money on projects in $T'$, $c$ must be in $T'$, i.e. they do have the virtual money to buy $T'$ but it would overshoot the budget. In the first case, for every project not in $T'$ that members of $S$ pay for, $|W'|$ grows by one (since they can only pay for projects they approve). In order to keep $|W'|<|T'|$, the mean amount of money they have paid at time $t$ for such a project must be greater than the mean cost of projects in $T'$. Otherwise, the number of projects they would pay  for (that they approve of and that are selected) would exceed the number of projects in $T'$. However, for each project not in $T'$ that voters from $S$ pay for, they should (as a group) pay less than the cost of any project from $T'$ not yet selected. Otherwise, they would have paid earlier for a cheaper project from $T'$. This is a contradiction. 
 Hence, the voters from $S$ only spent their money on projects from $T'$, and $c\in T'$. Let us assume that $c$ is the last project from $T'$ that is not yet selected.\footnote{If there are more projects from $T'$ not yet selected, we get that $x\leq \text{cost}(T')\cdot \frac{n}{|S|}$, so $x\leq 1$ still holds.} Because the rule stops exactly when $c$ can be paid by its supporters, we know that at that point in time, the voters in $S$ have earned exactly cost$(T')$ units of money, so $t\cdot |S| = \text{cost}(T')$. Hence, the total amount of money earned at time $t$ is $x=t\cdot n = \text{cost}(T')\cdot \frac{n}{|S|}$. Because $S$ is $(\alpha,T)$-cohesive, cost$(T')\leq \text{cost}(T)\leq \frac{|S|}{n}$, so 
 $
     x= \text{cost}(T')\cdot \frac{n}{|S|} \leq \frac{|S|}{n}\frac{n}{|S|}=1.
 $
 However, we also had that $x>1+y>1$. Contradiction.
\end{proof}
\begin{proposition}\label{prop:Phragmen-priceable}
    \hyperlink{link:prop:Phragmen-priceable}{Phragm\'{e}n satisfies priceability.}
\end{proposition}
\cite{Peters2020_lim_of_welf} show that Phragm\'{e}n does not satisfy EJR in MWV- and, therefore, PB-instances. Since the core implies EJR, Phragm\'{e}n does not satisfy the core neither in MWV- nor in PB-instances.
\begin{proposition}\label{prop:Phragmen-negative}
    \hyperlink{link:prop:Phragmen-negative}{Phragm\'{e}n and Rule X do not satisfy LP.}
\end{proposition}

\begin{proof}[Proof sketch]
LP requires any affordable unanimously approved project to be selected, but Phragm\'{e}n and Rule X do
not necessarily select those if their cost is high enough compared to the other projects.
\end{proof}

\begin{proposition}\label{prop:RuleX-PJR_PB}
    Rule X satisfies PJR-up-to-one.
\end{proposition}
\begin{proof} Rule X satisfies EJR-up-to-one \citep{peters2020proportionalPB}. Theorem \ref{thm:PB-EJR->PB-PJR} will show that EJR-up-to-one implies PJR-up-to-one. Hence Rule X also satisfies PJR-up-to-one.
\end{proof}


\section{Relations Between Axioms}
We study the logical relationships among the axioms we introduced and report on the results recapitulated in Figure \ref{fig:relations}.

\subsection{Priceability, PJR, EJR, and the Core}\label{sec:pr-PJR_PB}
We start by showing that PJR, EJR, and the core do not imply priceability in MWV-instances, even in laminar ones. 
\begin{theorem}\label{thm:PJR-pr-laminarinstances_MWV}
In laminar MWV-instances there exist bundles that satisfy PJR, EJR, or are in the core, but are not priceable.
\end{theorem}
Since laminar MWV-instances are a specific type of MWV-instances, which are a specific type of PB-instances, the same result holds for MWV and PB.

In MWV-instances, every priceable bundle satisfies PJR, as is shown by \citep[Prop. 1]{Peters2020_lim_of_welf}. 
This raises the question whether this relation is also present in the PB setting. We show now that this is not the case.
\begin{theorem}\label{thm:pr_not->PJR-PB}
    There are priceable bundles not satisfying PJR.
\end{theorem}

\begin{proof}[Proof sketch]
PJR is based on the utility of the voters being higher than some threshold $\alpha(c)$, while priceability only discriminates between utilities of 0 and utilities above zero. 
It is therefore possible to construct a PB-instance with a bundle that is priceable but does not satisfy PJR.
\end{proof}

From this result, it follows that priceability neither implies EJR nor the core. Priceability would otherwise imply PJR, since the core implies EJR, and EJR implies PJR (Theorem \ref{thm:PB-EJR->PB-PJR}).
These results hold even for MWV. 
 
\begin{theorem}\label{thm:PB-EJR->PB-PJR}
EJR(-up-to-one) implies PJR(-up-to-one).\end{theorem}

\subsection{Laminar Proportionality and Priceability}\label{sec:LP-pr_PB}
We first show that in MWV-instances, and therefore also in  PB-instances, priceability does not imply LP.
\begin{theorem}\label{th:pr_not_imply_lp}
In laminar MWV-instances there exist priceable bundles that are not laminar proportional.
\end{theorem}

However it is worth reporting that for a class of price systems (called balanced systems) we are able to prove that the implication goes through. 
Furthermore, LP implies priceability on laminar election instances even in the PB setting.

\begin{theorem}\label{thm:LP-pr_PB}
LP implies priceability in laminar PB-instances.
\end{theorem}
\begin{proof}[Proof sketch] 
By induction on the structure of laminar PB-instances. We show that for every bundle $W$ that is laminar proportional for a laminar PB-instance ($P,l$), there exists a price system $\mathbf{ps}=(b,(p_i)_{i\in N})$ where  $b=\text{cost}(W)$. 

Consider three cases. The first is the basis of the induction. 
If $P$ is unanimous with $\text{cost}(C(P))\geq l$ and $W$ is laminar proportional for ($P,l$) (with cost$(W)\leq l$), then $W\subseteq C(P)$, so voters can divide their budget over $W$. With an initial $b=\text{cost}(W)$, all and only the projects in $W$ can be bought.

In the second case, a unanimously approved project $c\in W$ such that $W'=W\backslash\{c\}$ is laminar proportional. By the inductive hypothesis, we know that there exists a price system $\mathbf{ps'}$ with initial budget $b'=\text{cost}(W')$. Because $c$ is unanimously approved, all voters could pay for $c$.
We know that in $\mathbf{ps'}$, there was no project not in $W'$ that was affordable to its supporters.
If every voter got $\frac{\text{cost}(c)}{n}$ more budget, to spend entirely on $c$, $c$ would get funded and no voter would have more unspent budget than before. 
Also, the initial budget of every voter is now $\frac{b'}{n}+\frac{\text{cost}(c)}{n} = \frac{\text{cost}(W')+\text{cost}(c)}{n} = \frac{\text{cost}(W)}{n}=\frac{b}{n}$ units of money, and the initial budget is $b=\text{cost}(W)$ and all the individual payment functions stay the same. Because for every project $c$ in $W'$ the sum of the individual payments was equal to cost$(c)$, this is also the case for every project in $W$.

In the third case, the profile consists of two laminar profiles:  $P=P_1+P_2$ and $l=l_1+l_2$, where $W_1$ and $W_2$ are laminar proportional for respectively $(P_1,l_1)$, $(P_2,l_2)$. Take $W=W_1\cup W_2$, which is by definition laminar proportional for $(P,l)$. By the inductive hypothesis, there exist price systems $\mathbf{ps_1}=(b_1,\{p_{1,i}\}_{i\in N})$ and $\mathbf{ps_2}=(b_2,\{p_{2,i}\}_{i\in N})$  with initial budgets $b_1=\text{cost}(W_1)$ and $b_2=\text{cost}(W_2)$, and that $|P_1|\cdot l_2 = |P_2|\cdot l_1$. We can now define a price system $\mathbf{ps}$ that supports $W$ as follows: $\mathbf{ps}=(b,(p_i)_{i\in N})$ with $b=\text{cost}(W) = b_1+b_2$, and for all voters $i\in N$,
    $ p_i(c)= p'_{1,i}(c) + p'_{2,i}(c),$ where $p'_{1,i}$ and $p'_{2,i}$ are extended versions of respectively $p_{1,i}$ and $p_{2,i}$ that yield zero for the projects that those are not defined for.
    It is easy to check that $\mathbf{ps}$ supports $W$ according to the criteria from Def. \ref{def:Priceability_PB}.
\end{proof}

\subsection{Laminar Proportionality and the Core} \label{sec:LP_core_PB}

\begin{theorem}\label{thm:EJR-lp_-laminarinstances_MWV}
In laminar MWV-instances, there exist bundles that satisfy PJR, EJR, or are in the core, but do not satisfy LP.
\end{theorem}

\begin{theorem} \label{thm:lp-core_PB}
There exist laminar proportional bundles that do not satisfy PJR, EJR, or are not in the core.
\end{theorem}
Theorem \ref{thm:lp-core_PB} holds because of instances where there is one relatively expensive unanimously approved project, that is included in any laminar proportional bundle, but where there are many cheap projects that can be satisfactory enough for a group of voters. 
We show now that under certain restrictions, laminar proportional bundles are in the core (and hence also satisfy EJR and PJR).
To do that we define a property of bundles which we call \textit{unanimity-affordability} (u-affordability):
\begin{definition}\label{def:P-u-afford}
 A bundle $T$ is u-affordable (shortly, u-afford) w.r.t. instance $(P,l)$ whenever
 for any unanimously approved project $c\in C(P)$ 
 there exists $t\in T$ s.t. $\text{cost}(t)\geq\text{cost}(c)$.
\end{definition}
Since $\text{cost}(c)\geq\text{cost}(c)$, the definition is also satisfied if $c\in T$. 
We will show that in laminar PB-instances, a laminar proportional bundle satisfies the core if it is {\em subject to u-afford}, that is, if u-afford holds for the bundle inductively on the structure of the instance.
\begin{theorem}\label{thm:LP-core_u-afford}
In laminar PB-instances, laminar proportional bundles subject to u-afford satisfy the core.
\end{theorem}
 \begin{proof}[Proof sketch]
The proof is by induction on laminar PB instances. Recall that a set of projects is in the core if for every group of voters $S\subseteq N$ and set of projects $T\subseteq C$ such that $S$ can afford $T$ with their share of the budget there is a voter $i\in S$ such that $u_i(W)\geq u_i(T)$. The proof then shows that if for any unanimously approved project $c$ in $W$ either $c$ is part of $T$, or there is some project in $T$ that costs at least as much as $c$, then there exists a voter $i\in S$ such that $u_i(W)\geq u_i(T)$.
The interesting case concerns the step with a unanimous candidate $c$: any group $S$ that could possibly block the core with bundle $T$ should all prefer $T$ over $W$ or not be able to buy $W$. This is impossible because by u-afford $T$ contains a project with a larger cost than $c$.
\end{proof}

Since the core implies EJR and PJR (as showed by \cite{peters2020proportionalPB} and Theorem \ref{thm:PB-EJR->PB-PJR}), any laminar profile also satisfies EJR and PJR under the same restrictions.
Note that in MWV-instances there always is a project in $T$ that has cost of at least cost$(c)$, so our restriction applies.
Hence, there, laminar bundles are in the core, and so satisfy EJR and PJR.
\begin{corollary}\label{cor:LP-PJR-EJR-core_MWV}
In laminar MWV-instances, LP implies PJR, EJR, and the core.
\end{corollary}


\section{Conclusions and Future Work}

Our study is a contribution towards the systematization of the axiomatic landscape of proportionality in PB (Figure \ref{fig:relations}, Table \ref{tab:rules_and_axioms_PB}). With respect to proportionality-inspired rules, we showed that priceability and PB generalizations of PJR and LP do not discriminate between Phgragm\'{e}n and Rule X (unlike EJR). 

We focused on the basic PB framework and on cardinal and approval-based utilities. Extending the study to richer PB settings (e.g., diversity constraints \citep{div_constraints1, div_constraints2, div_constraints3}, project groups \citep{Project_groups}, negative attitudes \citep{PB_with_neg_feelings}, or resource types \citep{aziz_shah_pb}), or the yet more general ordinal preferences setting \citep{aziz21proportionality}, are natural avenues of future research.

\section*{Acknowledgments}
The authors are greatly indebted to Jan Maly, who pointed out a mistake in an earlier version of the paper, and to the anonymous reviewers of IJCAI/ECAI'22 for many helpful comments. 
Zoé Christoff acknowledges support from the project Social Networks and Democracy (VENI project number Vl.Veni.201F.032) financed by the Netherlands
Organisation for Scientific Research (NWO). Davide Grossi acknowledges support by the \href{https://hybrid-intelligence-
centre.nl}{Hybrid Intelligence Center}, a 10-year program funded by the Dutch
Ministry of Education, Culture and Science through the Netherlands
Organisation for Scientific Research (NWO).


\bibliographystyle{named}
\bibliography{biblio}

\clearpage

\appendix
\section{Proofs of Theorems and Propositions}\label{sec:appendix-proofs}
\paragraph{Proposition \ref{prop:ejr-up-to-one}}
\begin{quote}\it
Def.~\ref{def:EJR_PB} and Def.~\ref{def:EJR_PB-up-to-one} are equivalent in MWV-instances.
\end{quote}
\begin{proof}
Since in MWV utilities are either 0 or 1, we can assume that for an $(\alpha,T)$-cohesive group of voters $S$, for each $c\in T$, $\alpha(c)=1$. Then $\sum_{c\in T}\alpha(c)=|T|$, and  $u_i(\mathcal{R}(E))\geq\sum_{c\in T}\alpha(c)$ translates to $|A_i\cap\mathcal{R}(E)|\geq |T|$. The ``up-to-one'' condition  ``for some $a\in C$,  $u_i(\mathcal{R}(E)\cup \{a\})>\sum_{c\in T}\alpha(c)$'' (note the strict inequality) is in MWV equivalent to ``for some $a\in C$, $|A_i\cap\mathcal{R}(E)\cup\{a\}|>|T|$'', which is, since we only count the size of a set in integers, also equivalent to $|A_i\cap\mathcal{R}(E)|\geq|T|$.
\end{proof}

\paragraph{Theorem \ref{thm:PB-PJR}}
\begin{quote} \it
    MWV-PJR and PJR are equivalent in MWV-instances.
\end{quote}

In order to show that given definition (Definition \ref{def:PB-PJR}) is indeed a generalisation of MWV-PJR, we apply it to the multi-winner voting setting and show that MWV-PJR and PJR are equivalent in MWV-instances. Recall the definition of MWV-PJR:
\begin{definition}[Proportional justified representation for MWV \citep{sanchezfernandez2016}]
 An approval based voting rule $\mathcal{R}$ satisfies proportional justified representation for MWV (MWV-PJR) if for every ballot profile $P$ and committee size $k$, the rule outputs a committee $W = \mathcal{R}(P,k)$ s.t.:
for every $\ell\leq k$ and every $\ell$-cohesive set of voters $S\subseteq N$, it holds that $|W\cap\left(\cup_{i\in S}A_i\right)|\geq \ell$, where a set $S$ is $\ell$-cohesive if $|S|\geq \ell\cdot \frac{n}{k}$ and $|\cap_{i\in S}A_i|\geq \ell$.
\end{definition}

We need the following lemma:

\begin{lemma} \label{lemma1}
Let $N$ be a set of agents, $C$ a set of alternatives, let each project $c$ have a cost  $\text{cost}(c)=\frac{1}{k}$, and let each voter $i$ have a utility $u_i(c)\in \{0,1\}$ for project $c$. Define the approval set $A_i$ of voter $i$ as $A_i=\{c\in C: u_i(c)=1\}$.
Then
\begin{enumerate}[(a)]
    \item for given $\alpha:C\rightarrow[0,1] $ and $T\subseteq C$,  any group of voters $S$ that is $(\alpha, T)$-cohesive is also $\ell$-cohesive for $\ell=|T'|$ with $T'= \{c\in T: \alpha(c)> 0\}$, and
    \item for every group $S$ that is $\ell$-cohesive there are $T\subseteq C$ with $|T|=\ell$ and $\alpha: C\rightarrow [0,1]$ with $\alpha(c)=1$ for all $c\in C$, such that $S$ is $(\alpha,T)$-cohesive.
\end{enumerate}
\end{lemma}
\begin{proof}[Proof of of Lemma \ref{lemma1}]
\textbf{(a)}: Assume that $S\subseteq N$ is $(\alpha, T)$-cohesive for some  $\alpha: C\rightarrow [0,1]$ and $T\subseteq C$. By definition, in the approval based multi-winner setting this means that $|S|\geq |T|\cdot \frac{n}{k}$ and that for all $i\in S$ and all $c\in T$, $u_i(c)\geq \alpha(c)$. Now take a subset $T'$ from $T$ of all the projects in $T$ that have at least some utility for all voters in $S$: $T'= \{c\in T: \alpha(c)> 0\}$. 
Because $u_i(c)\in \{0,1\}$ (i.e. voting is approval based), $\alpha(c)>0$ and $u_i(c)\geq \alpha(c)$ imply that $u_i(c)=1$. Hence, for all $i\in S$ and all $c \in T', u_i(c)=1$. 
We can rewrite this as $T'\subseteq \cap_{i\in S}A_i$, so  $|T'|\leq |\cap_{i\in S}A_i|$.
Now if we call $|T'|=\ell$, we have $\ell\leq |T|$, so $|S|\geq \ell\cdot \frac{n}{k}$ and $|\cap_{i\in S}A^\alpha_i|\geq \ell$
, so $S$ is $\ell$-cohesive.\\
\textbf{(b)}: Assume that $S\subseteq N$ is $\ell$-cohesive for some $\ell\leq k$. This means that $|S|\geq \ell \cdot \frac{n}{k}$ and $|\cap_{i\in S}A_i|\geq \ell$. We take $T\subseteq \cap_{i\in S}A_i$ with $|T|=\ell$ (which we can do because $\ell\leq |\cap_{i\in S}A_i|$). Now take $\alpha: C\rightarrow [0,1]$ such that $\alpha(c) =1$ for all $c\in C$. Then $|S|\geq|T|\cdot \frac{n}{k}$ and for all $i\in S$ and all $c \in T, u_i(c)=1=\alpha (c)$, so $S$ is $(\alpha, T)$-cohesive.
\end{proof}

\begin{proof}[Proof of Theorem \ref{thm:PB-PJR}]
We show that a rule that satisfies the new definition of PJR also satisfies the old MWV-PJR for every MWV-instance. 

\textbf{PJR $\Rightarrow$MWV-PJR}: Assume that a rule $\mathcal{R}$ satisfies PJR. Now take an arbitrary MWV-instance $E$ where all projects have unit-cost, and voting is approval based: the utility of a project for a voter is either 0 (when the voter does not approve the project) or 1 (when the voter does approve the project). According to PJR, it is true that for all $S$, $\alpha: C\rightarrow[0,1]$ ,  and $T$ with  $T\subseteq C$, if $S$ is $(\alpha,T)$-cohesive ($|S|\geq \text{cost}(T)\cdot n$ and it holds that $u_i(c)\geq \alpha (c)$ for every voter $i\in S$ and each project $c\in T$), then 
$$\sum_{c\in \mathcal{R}(E)}(\max_{i\in S}u_i(c))\geq\sum_{c\in T}\alpha(c).$$

Because $E$ satisfies the assumptions of unit-cost and approval based voting, this boils down to the following: for all $S, \alpha: C\rightarrow[0,1]$, and  $T\subseteq C$ with $|S| \geq |T|\cdot \frac{n}{k}$
and for which $u_i(c)\geq \alpha(c) $ for all $ i\in S$ and for all $c \in T$,
it is the case that
\begin{equation}
    |\mathcal{R}(E)\cap(\cup_{i\in S} A_i)|\geq\sum_{c\in T}\alpha(c).
\end{equation}

In order to show that $\mathcal{R}$ satisfies MWV-PJR, we have to prove that for all $S$ and $\ell\leq k$ where $S$ is $\ell$-cohesive ($|S|\geq \ell \cdot \frac{n}{k}$ and $|\cap_{i\in S}A_i|\geq l$), $$|\mathcal{R}(E)\cap(\cup_{i\in S} A_i)|\geq \ell.$$ 
Take arbitrary $S\subseteq N$ and $\ell\leq k$ and suppose that $S$ is $\ell$-cohesive. According to Lemma 1(b), there are $T\subseteq C$ with $|T|=\ell$ and $\alpha: C\rightarrow [0,1]$ with $\alpha(c)=1$ for all $c\in C$, such that $S$ is $(\alpha,T)$-cohesive. Because $\mathcal{R}$ satisfies PJR, this implies that 
$|\mathcal{R}(E)\cap(\cup_{i\in S} A_i)|\geq\sum_{c\in T}\alpha(c)$. However, because of our choice of $T$ and $\alpha$, we know that $\sum_{c\in T}\alpha(c)=|T|=\ell$, so $|\mathcal{R}(E)\cap(\cup_{i\in S} A_i)|\geq \ell$, which is what we had to prove. \\
\\
We still have to prove that a rule that satisfies MWV-PJR also satisfies the newly defined PJR in MWV-elections:

\textbf{MWV-PJR $\Rightarrow$PJR}:
Assume that a rule $\mathcal{R}$ satisfies MWV-PJR: in every election instance $E$, for every $\ell$-cohesive group of voters $S$, $|\mathcal{R}(E)\cap(\cup_{i\in S} A_i)|\geq \ell.$ We want to prove that in every MWV-instance, any ($\alpha, T$)-cohesive group $S$ satisfies $|\mathcal{R}(E)\cap(\cup_{i\in S} A_i)|\geq\sum_{c\in T}\alpha(c)$. Take arbitrary such election instance $E$, and suppose that a group $S$ is $(\alpha,T)$-cohesive. Then according to Lemma 1(a), when we take $T'= \{c\in T: \alpha(c)> 0\}$, $S$ is $\ell$-cohesive for $\ell=|T'|$. 
Because $R$ satisfies MWV-PJR, it follows that $|\mathcal{R}(E)\cap(\cup_{i\in S} A_i)|\geq \ell = |T'|$. By definition of $T'$, $|T'|\geq \sum_{c\in T'}\alpha(c) = \sum_{c\in T}\alpha(c)$, so $|\mathcal{R}(E)\cap(\cup_{i\in S} A_i)|\geq\sum_{c\in T}\alpha(c)$, which is what we had to prove. 
\end{proof}
\newpage
\paragraph{Proposition \ref{prop:PAV-PJR}}
\begin{quote}\it
    PAV does not satisfy PJR.
\end{quote}
\begin{proof}As \cite{peters2020proportionalPB_arxiv} show by the example of Onetown, PAV does guarantee proportional representation, so in specific it does not satisfy PJR. Because PAV uses approval voting, we should use  Definition \ref{def:PJR_PB_approval} for PJR here. The group of voters that live in Leftside are $T$-cohesive for $T=\{L_1, L_2, L_3\}$: they can with their share of the money afford all projects in $T$ and do all approve all projects in $T$. However, the amount of projects in the committee $W$ that PAV returns that at least one of the voters in Leftside ($S$) approves of is $|W\cap_{i\in S}A_i| =2$, which is less than the number of projects in $T$.
\end{proof}

\paragraph{\hypertarget{link:prop:Phragmen-priceable}{Proposition \ref{prop:Phragmen-priceable}}}
\begin{quote}\it
    Phragm\'{e}n satisfies priceability.
\end{quote}
\begin{proof}As \cite{peters2020proportionalPB} mention, there is no obvious generalization of Phragm\'{e}n to non-approval voting. For non-unit costs, it is however directly applicable. We can easily see that Phragm\'{e}n for non-unit costs is still priceable (with the PB definition of priceability). We can construct a price system as follows: If the rule stopped at time $t$, let the initial budget of every voter be $t$, so the total initial budget is $b=t\cdot n$. Then every time that a project is selected by Phragm\'{e}n and added to the winning bundle $W$, add the amount of money that a voter $i$ pays for it to $p_i(c)$. Clearly, in the price system constructed in this way, the cost for every project $c\in W$ is paid, every voter has spend at most $\frac{b}{n}$ units of money and pays only for voters she approves, and for each non-selected project $c\notin W$, its supporters do not have enough money to buy it, because if they would have, $c$ would have been added to $W$ already.
\end{proof}
 
 \paragraph{Proposition \hypertarget{link:prop:Phragmen-negative}{\ref{prop:Phragmen-negative}}}
 \begin{quote}\it
    Phragm\'{e}n and Rule X do not satisfy LP.
 \end{quote}
 \begin{proof} As a counterexample, we can use Table \ref{tab:counterexample2}:
 \begin{table}[t]
    \centering
        \setlength{\tabcolsep}{20pt}
    \begin{tabular}{c c c }
    \cline{1-2}
        \multicolumn{2}{|c|}{\cellcolor[HTML]{C1C1C1}$c_3$, 0.5} &   \\ \hline
        \multicolumn{2}{|c|}{\cellcolor[HTML]{C1C1C1}$c_2$, 0.5} &   \multicolumn{1}{c|}{$c_5$, 0.5}  \\ \hline
        \multicolumn{2}{|c|}{\cellcolor[HTML]{C1C1C1}$c_1$, 0.5} &   \multicolumn{1}{c|}{\cellcolor[HTML]{C1C1C1}$c_4$, 0.5} \\ \hline
        \multicolumn{3}{|c|}{$c_6$, 0.7}  \\ \hline
        $v_1$ & $v_2$ & $v_3$ 
    \end{tabular}
    
    \caption{A counterexample that shows that Phragm\'{e}n and Rule X do not satisfy LP and that priceability does not imply LP: profile $P$ with bundle $W$ (in grey). It should be read in the same way as Table \ref{tab:ex_lamprop}}
    \label{tab:counterexample2}
\end{table}
Suppose the approval votes are as shown in Table \ref{tab:counterexample2}, projects $c_1, ..., c_5$ have a cost of 0.1 and project $c_6$ has a cost of 0.7, and the total budget is 1. The profile is still a laminar PB-instance according to Definition \ref{def:laminar_PB}. In this situation, Phragm\'{e}n will return $\{c_1, ..., c_5\}$: at time $t=0.05$, $v_1$ and $v_2$ can buy $c_1$, then at $t=0.1$, they can buy $c_2$ and $v_3$ can buy $c_4$, at $t=0.15$, $v_1$ and $v_2$ can buy $c_3$ and at $t=0.2$, $v_3$ can buy $c_5$. The remaining amount of budget is 0.5, so $c_6$ is not affordable anymore, and $W=\{c_1, ..., c_5\}$ is returned. The same bundle will be returned for Rule X: each voter starts with a budget of $\frac{1}{3}$. $c_1, c_2$, and $c_3$ are $\rho$-affordable for $\rho=0.05$, $c_4$ and $c_5$ are $\rho$-affordable for $\rho=0.1$ (whereas $c_6$ would only be $\rho$-affordable for $\rho\geq \frac{0.7}{3}$ in the beginning), the remaining budget is not enough to buy $c_6$, so $W$ will be returned. A laminar proportional bundle, however, would consist of $c_6$, two of $\{c_1, c_2, c_3\}$ and one of $\{c_4, c_5\}$.
\end{proof}

\paragraph{Theorems \ref{thm:PJR-pr-laminarinstances_MWV} and \ref{thm:EJR-lp_-laminarinstances_MWV}}
\begin{quote}\it
    In laminar MWV-instances there exist bundles that satisfy PJR or EJR or are in the core, but are not priceable. (3)\\
    In laminar MWV-instances, there exist bundles that satisfy PJR, EJR, or are in the core, but do not satisfy LP. (8)
\end{quote}
\begin{proof} 
\begin{table}[t]
    \centering
        \setlength{\tabcolsep}{20pt}
    \begin{tabular}{c c c }
    \cline{1-2}
        \multicolumn{2}{|c|}{\cellcolor[HTML]{C1C1C1}$c_2$} & \multicolumn{1}{c|}{\cellcolor[HTML]{C1C1C1}$c_4$}  \\ \hline
        \multicolumn{2}{|c|}{\cellcolor[HTML]{C1C1C1}$c_1$} &   \multicolumn{1}{c|}{\cellcolor[HTML]{C1C1C1}$c_3$} \\ \hline
        \multicolumn{3}{|c|}{$c_5$}  \\ \hline
        $v_1$ & $v_2$ & $v_3$ 
    \end{tabular}
    \caption{A counterexample that shows that PJR, EJR, and the core do not imply priceability or LP: profile P with bundle W indicated in grey. It should be read in the same manner as Table \ref{tab:ex_lamprop}, the costs are the same for all projects.
    }
    \label{tab:counterexample_PJR_priceability}
\end{table}We start by giving a counterexample (shown in Table \ref{tab:counterexample_PJR_priceability}) that shows that a bundle $W$ in a laminar profile can satisfy PJR without being priceable. Suppose we have 3 voters, $v_1, v_2,$ and $v_3$, and 5 projects $c_1, ..., c_5$, and that $k=4$. Projects $c_1$ and $c_2$  are approved by  voters $v_1$ and $v_2$, projects $c_3$ and $c_4$ are approved by voter $v_3$, and project $c_5$ is approved by all three voters. This election instance is laminar, as can easily be checked, and the selected bundle $W$, as indicated in grey in Table \ref{tab:counterexample_PJR_priceability} satisfies PJR: for all $\ell \leq 4$ and all $\ell$-cohesive group of voters $S$, it holds that $|W\cap \cup_{i\in S}A_i|\geq \ell$. It is, however, not priceable: suppose, for a contradiction, that there is a price system $\mathbf{ps}$ that supports $W$. The price $p$ of this system has to be such that $v_3$ can pay for both $c_3$ and $c_4$ (because no other voter can pay for these projects), so $p\leq 0.5$. However, $v_1$ and $v_2$ together have a budget of 2, which they must spend only on $c_1$ and $c_2$. Hence, $p>\frac{2}{3}$, because otherwise $v_1$ and $v_2$ together have enough unspent budget to pay for $c_5$ which is not in $W$. Hence $p\leq 0.5$ \textit{and} $p>\frac{2}{3}$, which is a contradiction. Therefore, there is no price system that supports $W$.
 
This counterexample also shows that a committee that satisfies EJR or is in the core does not necessarily have to be priceable or laminar proportional. The committee in Table \ref{tab:counterexample_PJR_priceability} does satisfy EJR and it is in the core, as can easily be checked, but it is neither priceable nor laminar proportional. 
\end{proof}
 
\paragraph{Theorem \ref{thm:pr_not->PJR-PB}} 
\begin{quote}\it
    There exist priceable bundles that do not satisfy PJR.
\end{quote}
 \begin{proof} 
Take a PB-instance $E$ with $N=\{s_1, s_2, v_1, v_2, v_3\}$, $C=\{t_1, t_2, c_1, c_2, c_3\}$, the cost of all projects is 0.2, $\alpha(t_1)=\alpha(t_2)=0.4$, voters $s_1$ and $s_2$ have some utility only for the $t$-projects: $u_{s_1}(t_1)=u_{s_1}(t_2)=0.6, u_{s_1}(c)=0$ for $c\in \{c_1, c_2, c_3\}$, $u_{s_2}=u_{s_1}$, and voters $v_1, v_2$, and $v_3$ only have utility for the $c$-projects: for $v\in \{v_1, v_2, v_3\}$, $u_v(t_1)=u_v(t_2)=0$ and $u_v(c)>0$ for $c\in \{c_1, c_2, c_3\}$. Furthermore, define a pricesystem ps with $b=1$ and $p_s(c)=0$ for $c\in \{c_1, c_2, c_3, t_1\}$, and $p_s(t_2)=0.1$ for $s\in \{s_1, s_2\}$, and with $p_v(t)=0$ and $p_v(c)=\frac{0.2}{3}$ for $v\in \{v_1, v_2, v_3\}$, $t\in \{t_1, t_2\}$, and $c\in \{c_1, c_2, c_3\}$. 

Now, let bundle $W=\{t_2, c_1, c_2, c_3\}$ be the winning bundle of some election rule.
We have: 
\begin{enumerate}[(1)]
    \item A voter can only pay for projects she gets at least some utility form: if $u_i(c)=0$, then $p_i(c)=0$ for each $i\in N$ and $c\in C$;
    \item Each voter can spend the same budget of $\frac{b}{n}$ units of money: $\sum_{c\in C}p_i(c)\leq \frac{1}{n}$ for each $i\in N$;
    \item for each $c \in W $, the sum of the payments for $c$ equals its price: $\sum_{i\in N}p_i(c)=cost(c)$;
    \item no project outside of the bundle gets any payment: for all $c\notin W, \sum_{i\in N}p_i(c)=0$;
    \item there exists no non-selected project whose supporters in total have a remaining unspent budget of more than its cost: for all $c\notin W, \sum_{i\in N \text{ for which } u_i(c)>0}(b-\sum_{c'\in W}p_i(c'))\leq cost(c)$.
\end{enumerate}
Hence, $W$ is a priceable bundle. However, if we take $S=\{s_1, s_2\}$ and $T=\{t_1, t_2\}$, we have $|S|=2=cost(T)\cdot n$ and $u_i(c)\geq \alpha (c)$ for all $i\in S, c\in T$, so $S$ is $(\alpha, T)$-cohesive. Nevertheless, $\sum_{c\in W}(\max_{i\in S}u_i(c))=0.6<0.8=\sum_{c\in T}\alpha(c)$, s$W$ does not satisfy PJR. Therefore, this counterexample shows that it is not the case that all priceable bundles satisfy PJR. 
\end{proof}

\paragraph{Theorem \ref{thm:PB-EJR->PB-PJR}}
\begin{quote}\it
    EJR(-up-to-one) implies PJR(-up-to-one).
\end{quote}

\begin{proof} 
We prove the implication with the up-to-one condition. The implication without the condition follows directly. Suppose that rule $\mathcal{R}$ satisfies EJR-up-to-one and take an $(\alpha,T)$-cohesive group of voters $S$ for 
some $\alpha: T\rightarrow [0,1], T\subseteq C$. Because $\mathcal{R}$ satisfies EJR-up-to-one, there is a voter $i\in S$ such that either (1) $u_i(\mathcal{R}(E))\geq\sum_{c\in T}\alpha(c)$ or (2) for some $a\in C$ it holds that $u_i(\mathcal{R}(E)\cup \{a\})>\sum_{c\in T}\alpha(c)$. For this voter $i$, in case (1),  $\sum_{c\in \mathcal{R}(E)}(u_i(c))\geq\sum_{c\in T}\alpha(c)$. In case (2), $\sum_{c\in \mathcal{R}(E)\cup \{a\}}(u_i(c))\geq\sum_{c\in T}\alpha(c)$. Hence, $\mathcal{R}$ satisfies PJR-up-to-one, since $\sum_{c\in \mathcal{R}(E)}(\max_{i\in S}u_i(c))\geq\sum_{c\in T}\alpha(c)$ or $\sum_{c\in \mathcal{R}(E)\cup\{a\}}(\max_{i\in S}u_i(c))\geq\sum_{c\in T}\alpha(c)$.
\end{proof}

\paragraph{Theorem \ref{th:pr_not_imply_lp}}
\begin{quote}\it
    In laminar MWV-instances, priceability does not imply LP.
\end{quote}
\begin{proof} Let us recall the definitions of laminar MWV-instances and laminar proportionality from \cite{Peters2020_lim_of_welf}.
\begin{definition}[Laminar MWV-instances] \label{def:Laminar_instances_MWV}
   An MWV-instance $(P,k)$ is \textit{laminar} if either:
    \begin{enumerate}
        \item $P$ is unanimous and $|C(P)|\geq k$.
        \item There is a candidate $c\in C(P)$ such that $c\in A_i$ for all $A_i\in P$, the profile $P_{-c}$ is not unanimous and the instance $(P_{-c}, k-1)$ is laminar (with $P_{-c}=(A_1\backslash \{c\}, ..., A_n\backslash\{c\})$).
        \item There are two laminar MWV-instances $(P_1, k_1)$ and $(P_2, k_2)$ with $C(P_1)\cap C(P_2)=\emptyset$ and  $|P_1|\cdot k_2 = |P_2|\cdot k_1$ such that $P=P_1+P_2$ and $k=k_1+k_2$.
    \end{enumerate}
    \end{definition}
    \begin{definition}[LP for MWV-instances]  \label{def:Laminar_Proportionality_MWV}
     A rule $\mathcal{R}$ satisfies LP if for every laminar MWV-instance with ballot profile $P$ and committee size $k$, $\mathcal{R}(P,k) = W$ where $W$ is a laminar proportional committee, i.e. 
    \begin{enumerate}
        \item If $P$ is unanimous, then $W \subseteq A_i$ for some $i\in N$ (if everyone agrees, then part of the candidates they agree on is chosen. 
        \item If there is a unanimously approved candidate $c$ s.t. $(P_{-c}, k-1)$ is laminar, then $W = W'\cup \{c\}$ where $W'$ is a committee which is laminar proportional for $(P_{-c}, k - 1)$.
        \item If $P$ is the sum of laminar instances $(P_1, k_1)$ and $(P_2, k_2)$, then $W = W_1 \cup W_2$ where $W_1$ is
laminar proportional for $(P_1, k_1)$ and $W_2$ is laminar proportional for $(P_2, k_2)$.
    \end{enumerate}
    \end{definition}
In order to show that priceability does not imply LP, we construct a counterexample. Consider a profile $P$ as in Table \ref{tab:counterexample2}, but where every project has unit cost. 
We have 3 voters, $v_1, v_2,$ and $v_3$, and 6 projects $c_1, ..., c_6$. Projects $c_1, c_2,$ and $c_3$ are approved by  voters $v_1$ and $v_2$, projects $c_4$ and $c_5$ are approved by voter $v_3$, and project $c_6$ is approved by all three voters.  This profile is laminar for $k=4$, we can construct it according to the inductive Definition \ref{def:Laminar_instances_MWV} by first concatenating $\{c_1, c_2, c_3\}$ ($k=2$) and $\{c_4, c_5\}$ ($k=1$) by Definition \ref{def:Laminar_instances_MWV}.3 and then adding $c_6$ by Definition \ref{def:Laminar_instances_MWV}.2, which results in $k=2+1+1=4$.\\ The selected bundle $W$, as indicated in grey, is not laminar proportional, because in order to be laminar proportional, $c_6$ would have to be included in $W$. It is, however, priceable: we can construct a price system with price $p=0.65$: $v_1$ and $v_2$ together pay for $c_1, c_2,$ and $c_3$ and have $2-3\cdot0.65 =0.05$ left over, and $v_3$ pays for $c_4$ and has $1-0.65=0.35$ left over. Hence, $v_3$ cannot pay for $c_5$ anymore, and all voters together have an unspent budget of $0.4$, so they cannot pay for $c_6$.
\end{proof}

\paragraph{Theorem \ref{thm:LP-pr_PB}}
\begin{quote}\it
    LP implies priceability. 
\end{quote}
\begin{proof} We construct an inductive proof on the structure of laminar PB-instances, very similar to the proof for the unit-cost case, to prove that for every bundle $W$ that is laminar proportional for a laminar PB-instance ($P,l$), where $P$ is the list of approval sets of the voters and $l$ is the budget, there exists a price system $\mathbf{ps}=(b,(p_i)_{i\in N})$ where  $b=\text{cost}(W)$. 

\textbf{Basis}: If $P$ is unanimous with $\text{cost}(C(P))\geq l$ and $W$ is laminar proportional for ($P,l$) (with cost$(W)\leq l$), then $W\subseteq C(P)$, so the voters can just divide their budget over the projects in $W$. If we set the initial budget to be $b=\text{cost}(W)$, every voter can spend $\frac{b}{n} =\frac{\text{cost}(W)}{n}$. We can now let every voter spend $\frac{\text{cost}(c)}{n}$ on every project $c\in W$, so every project $c\in W$ gets exactly cost$(c)$. Then every voter spends in total $\sum_{c\in W}\frac{\text{cost}(c)}{n} = \frac{\text{cost}(W)}{n}$, so does not have anything left to spend on other projects.

\textbf{Inductive Hypothesis}: Suppose that $(P',l')$, $(P_1,l_1)$, and $(P_2,l_2)$ are laminar PB-instances, bundles $W', W_1$, and $W_2$ are laminar proportional for respectively $(P',l')$, $(P_1,l_1)$, and $(P_2,l_2)$, and suppose that for $W'$ there exists a price system $\mathbf{ps'}$ with initial budget $b'=\text{cost}(W')$, for $W_1$ there exists a price system \textbf{ps$_1$} with initial budget $b_1=\text{cost}(W_1)$ and for $W_2$ there exists a price system \textbf{ps$_2$} with initial budget $b_2=\text{cost}(W_2)$. Furthermore, suppose that $P'$ is not unanimous, that $C(P_1)\cap C(P_2)=\emptyset$ and that $|P_1|\cdot l_2 = |P_2|\cdot l_1$.

\textbf{Inductive step}:
\begin{itemize}
    \item There is a unanimously approved project $c$ such that $P = P'_{+c}$, where $P'_{+c}=(A_1\cup\{c\}, ... A_n\cup\{c\})$ (case 2 of Definition \ref{def:laminar_PB}). Suppose that $W$ is laminar proportional for $(P,l'+\text{cost}(c))$, then $W=W'\cup \{c\}$. By the inductive hypothesis, there exists a price system $\mathbf{ps'}$ for $W'$ with initial budget $b'=\text{cost}(W')$. Because $c$ is unanimously approved, in theory all voters can pay for $c$. We know that in $\mathbf{ps'}$, there was no project that was not in $W'$ for which its supporters together had enough (more than its cost) unspent budget. If we would give every voter $\frac{\text{cost}(c)}{n}$ more budget, which we let them spend entirely on $c$, $c$ will get enough money and no voter will have more unspent budget than they had before. Also, the initial budget of every voter is now $\frac{b'}{n}+\frac{\text{cost}(c)}{n} = \frac{\text{cost}(W')+\text{cost}(c)}{n} = \frac{\text{cost}(W)}{n}=\frac{b}{n}$ units of money, and the initial budget is $b=\text{cost}(W)$ and all the individual payment functions stay the same. Because for every project $c$ in $W'$ the sum of the individual payments was equal to cost$(c)$, this is also the case for every project in $W$.
    
    Formally we define the price system $\mathbf{ps}$ for the instance $(P,l)$ as follows: $\mathbf{ps}=(b,(p_i)_{i\in N})$ with $b=\text{cost}(W)$ and $p_i: C\rightarrow [0,1]$ such that $p_i(c)=\frac{\text{cost}(c)}{n}$ and $p_i(d)=p'_i(d)$ for all other projects $d\in C(P)$, where $p'_i$ is the payment function of voter $i$ in the price system $\mathbf{ps'}$.
    
    To show that this is indeed a valid price system that supports $W$, we look at the five points of the definition of a price system that supports a bundle:
    \begin{enumerate}
        \item Voters only pay for projects they get at least some utility from because they did so in $ps'$, and the only project which they now pay for that they did not pay for before is $c$, which is unanimously approved, so has some utility for all $i\in N$.
        
        \item All voters $i\in N$ have an initial budget of $\frac{b}{n}$:
        \begin{eqnarray}
        \sum_{d\in C}p_i(d) &=& \sum_{d\in C}p'_i(d)+\frac{\text{cost}(c)}{n}\\
        &\leq & \frac{b'}{n}+ \frac{\text{cost}(c)}{n}\label{eq:1}\\
        &=& \frac{b}{n}\label{eq:2},
        \end{eqnarray}
        where (\ref{eq:1}) follows from the inductive hypothesis: because \textbf{ps$'$} is a price system that supports $W'$, the sum of the payments of voter $i$ for the items in $W'$ is smaller than or equal to $\frac{b}{n}$. Equation \ref{eq:2} holds because the new budget $b$ is defined as $b=\text{cost}(W)=\text{cost}(W'\cap \{c\})=b'+\text{cost}(c)$.
        
        \item For each selected project $d\in W$, if $d\neq c$ the sum of the payments is 
        \begin{eqnarray}
        \sum_{i\in N}p_i(d)&=& \sum_{i\in N}p'_i(d)\\
        &=&  \text{cost}(d).
        \end{eqnarray}
         This follows from the inductive hypothesis: because \textbf{ps$'$} is a price system that supports $W'$, the sum of the payments of all voters for $d$ equals its cost. 
         For $c$, $\sum_{i\in N}p_i(c)= n\cdot \frac{\text{cost}(c)}{n}=\text{cost}(c)$.

        \item For any non-selected project $d\notin W$, $ \sum_{i\in N}p_i(d) = \sum_{i\in N}p'_i(d) = 0$.
        \item For any project outside of the bundle $d\notin W$, its supporters do not have a remaining unspent budget of more than $\text{cost}(c)$:
        \begin{eqnarray}
           & & \sum_{i\in N \text{ for which } u_i(d)>0}(b-\sum_{e\in W=W'\cup \{c\}} p_i(e))\\
            & = &\sum_{i\in N: u_i(d)>0}(b-p_i(c)-\sum_{e\in W'} p_i(e))\\
            & =&\sum_{i\in N: u_i(d)>0}(\frac{\text{cost}(W)}{n}-\frac{\text{cost}(c)}{n}-\sum_{e\in W'} p'_i(e))\\
            & =&\sum_{i\in N: u_i(d)>0} (\frac{\text{cost}(W')}{n} - \sum_{e\in W'} p'_i(e))\\
            & =& \sum_{i\in N: u_i(d)>0} (b' - \sum_{e\in W'} p'_i(e)) \label{eq:3}\\
            &\leq& \text{cost}(d),
        \end{eqnarray}
    so there is no non-selected project whose supporters in total have a remaining unspent budget of more than its cost. Equation (\ref{eq:3}) follows from the inductive hypothesis because $\mathbf{ps'}$ is a price system that supports $W'$, so satisfies Definition \ref{def:Priceability_PB}.5, the other equations are just rewritings of the formula.
    \end{enumerate}
    Hence, $\mathbf{ps}$ is indeed a valid price system that supports bundle $W$.
    
    \item $P=P_1+P_2$ and $l=l_1+l_2$ (case 3 of Definition \ref{def:laminar_PB}). Take $W=W_1\cup W_2$, which is by definition laminar proportional for $(P,l)$. We have to show that $W$ is priceable for this election instance. Note that there are no overlapping projects between $P_1$ and $P_2$, there is no voter in $P_1$ that gets any utility from a project from $C(P_2)$, and no voter in $P_2$ that gets any utility from a project from $C(P_1)$. 
    By the inductive hypothesis,   there exists a price system $\mathbf{ps_1}=(b_1,\{p_{1,i}\}_{i\in N})$ for $W_1$ with initial budget $b_1=\text{cost}(W_1)$, and for $W_2$ there exists a price system $\mathbf{ps_2}=(b_2,\{p_{2,i}\}_{i\in N})$ with  $b_2=\text{cost}(W_2)$. Also by the inductive hypothesis, $|P_1|\cdot l_2 = |P_2|\cdot l_1$.
    
    We can now define a price system $\mathbf{ps}$ that supports $W$ as follows: $\mathbf{ps}=(b,(p_i)_{i\in N})$ with $b=\text{cost}(W) = b_1+b_2$, and for all voters $i\in N$,
    $$ p_i(c)= p'_{1,i}(c) + p'_{2,i}(c),$$ where $p'_{1,i}$ and $p'_{2,i}$ are extended versions of respectively $p_{1,i}$ and $p_{2,i}$ that yield zero for the candididates that those are not defined for: 
    \begin{displaymath}
        p'_{1,i}(c) = \left\{
        \begin{array}{ll}
        p_{1,i}(c) & \mbox{if $c\in C(P_1)$ and $i\in P_1$}; \\
        0 & \mbox{if $c\in C(P_2)$ or $i\in P_2$}.
        \end{array}
        \right.
    \end{displaymath}
    \begin{displaymath}
        p'_{2,i}(c) = \left\{
        \begin{array}{ll}
        0 & \mbox{if $c\in C(P_1)$ or $i\in P_1$}; \\
        p_{2,i}(c) & \mbox{if $c\in C(P_2)$ and $i\in P_2$}.
        \end{array}
        \right.
    \end{displaymath}

    Again, we show that this is a valid price system that supports $W$ by looking at the five points of the definition:
    \begin{enumerate}
        \item We know that $\mathbf{ps_1}$ is a valid price system that supports $W_1$, so for voters $i\in P_1$ and projects $c\in W_1$, if $p_{1,i}(c)>0$, then $u_i(c)>0$, so  $c\in A_i$. Analogously, for voters $i\in P_2$ and $c\in W_2$ if $p_{2,i}(c)>0$, then $u_i(c)>0$, so $c\in A_i$. 
        Suppose $p_i(c)>0$. If $c\in C(P_1)$, then $p_i(c) = p_{1,i}(c)$, so $i\in P_1$ because there is no voter in $P_2$ that approves a project from $C(P_1)$ and vice versa. Hence, for $c\in C(P_1)$, if $p_i(c)>0$, then $p_{1,i}(c)>0$ and then $c\in A_i$. Similarly, we can argue that for $c\in C(P_2)$, if $p_i(c)>0$, then $p_{2,i}(c)>0$ and then $c\in A_i$. Because $P=P_1+P_2$, $C(P)=C(P_1)\cup C(P_2)$, so for all $c\in C(P)$, if $p_i(c)>0$ then $c\in A_i$, so $u_i(c)>0$.
        \item 
        $\sum_{c\in C(P)}p_i(c) = \sum_{c\in C(P)}p'_{1,i}(c) + p'_{2,i}(c)$. 
        We already saw that voters from $P_1$ do not pay for projects from $C(P_2)$ and vice versa. Hence, if $i\in P_1$, then $\sum_{c\in C(P)}p_i(c) = \sum_{c\in C(P)}p'_{1,i}(c) \leq \frac{b_1}{n}$ by the inductive hypothesis (because \textbf{ps$_1$} is a valid price system with initial budget $b_1$). Furthermore, we have $\frac{b_1}{n}\leq\frac{b_1+b_2}{n}=\frac{b}{n}$, so $\sum_{c\in C(P)}p_i(c)\leq \frac{b}{n}$. If $i\in P_2$, then  $\sum_{c\in C(P)}p_i(c) = \sum_{c\in C(P)}p'_{2,i}(c)
        \leq \frac{b}{n}$, in the same way.
        \item For each selected project $c\in W$, the sum of its payments is $\sum_{i\in N}p_i(c)=\sum_{i\in N}(p'_{1,i}(c) + p'_{2,i}(c)) $. For $c\in C(P_x)$ (with $x\in\{1,2\}$) this is $\sum_{i\in N}p'_{x,i}(c) = \text{cost}(c) $. This follows from the inductive hypothesis that \textbf{ps$_1$} and \textbf{ps$_2$} are price systems that support $W_1$ and $W_2$, so the sum of payments of all voters in these systems for a selected project is equal to the cost of the project.
        \item Because $W = W_1\cup W_2$, any project that is not selected in the new bundle, $c\notin W$, was not selected in $W_1$ or $W_2$, so did not get any payment there: for $c\in C(P_x)$, $\sum_{i\in N}p_{x,i}(c)=0$. Hence, it also does not get any payment in the new system: for $c\in C(P_x)$,  $\sum_{i\in N}p_{i}(c)=\sum_{i\in N}p'_{x,i}(c)=\sum_{i\in N}p_{x,i}(c)=0$ (for $x\in\{1,2\}$).
        \item All non-selected projects are only supported by voters from their own `old' system, who did not have in total a remaining unspent budget of more than its cost there, so neither will they have it now:\\
        Without loss of generality, assume that an non-selected project $c\notin W$ is part of $C(P_1)$. Then because $c\notin W$, we also have $c\notin W_1$, because if it was in $W_1$, it would also have been in $W$. Because $\mathbf{ps_1}$ is a price system that supports $W_1$, we know that $\sum_{i\in N \text{ for which } u_i(c)>0}(1-\sum_{e\in W_1} p_{1,i}(e))\leq \text{cost}(c)$. 
        However, for all voters $i \in N \text{ for which } c\in A_i$, we have $i\in P_1$, so for all $e\in W_1, p_{1,i}(e)=p_i(e)$, and for all $e\in W_2, p_{i}(e)=0$. This implies that
        \begin{eqnarray*}
                \sum_{i\in N \text{ for which } u_i(c)>0}(1-\sum_{e\in W_1} p_{1,i}(e))&\leq& \text{cost}(c) \\
                &\Leftrightarrow&\\ 
                 \sum_{i\in N: u_i(c)>0}(1-\sum_{e\in W=W_1\cup W_2} p_i(e)) &\leq & \text{cost}(c).
        \end{eqnarray*}
        We can analogously show the same for $c\in C(P_2)$, so conclude that for all $c\in C(P)=C(P_1)\cup C(P_2)$, if $c\notin W$, $$\sum_{i\in N: u_i(c)>0}(1-\sum_{e\in W} p_i(e)) \leq \text{cost}(c)$$
    \end{enumerate}
    By these five points, we have shown that $\mathbf{ps}$ is indeed a valid price system that supports bundle $W$.
\end{itemize}
We have shown by induction over laminar PB-instances that, if a bundle $W$ is laminar proportional in a laminar PB-instance ($P,l$), it is also supported by a price system with with $b=\text{cost}(W)$. Hence we can conclude that LP implies priceability in laminar PB-instances.
\end{proof}

\paragraph{Theorem \ref{thm:lp-core_PB}}
\begin{quote}\it
    There exist laminar proportional bundles that do not satisfy PJR, EJR, or are not in the core.
\end{quote}
\begin{proof}
We will prove this theorem by giving a counterexample. Consider a situation with $N=\{v_1, v_2, v_3, v_4\}$, a unanimously approved project $c$ with cost$(c) = 1$, a set of 8 projects $T=\{t_1, ..., t_8\}$ that cost $\frac{1}{3}$ each and are all approved by $v_1, v_2,$ and $v_3$, and a set of 4 projects $\{x_1, x_2, x_3, x_4\}$ that also cost $\frac{1}{3}$ and are approved by $v_4$. This profile is shown in Figure \ref{fig:counter_lp-core}.
\begin{figure}[h]
\centering
    \begin{lrbox}{\mytable}
        \setlength{\tabcolsep}{20pt}
    \begin{tabular}{c c c c}
    \cline{1-3}
        \multicolumn{3}{|c|}{$t_8$, \tiny $\frac{1}{3}$} &  \\ \cline{1-3}
        \multicolumn{3}{|c|}{$t_7$, \tiny $\frac{1}{3}$} &  \\ \cline{1-3}
        \multicolumn{3}{|c|}{\cellcolor[HTML]{C1C1C1}$t_6$, \tiny $\frac{1}{3}$} &  \\ \cline{1-3}
        \multicolumn{3}{|c|}{\cellcolor[HTML]{C1C1C1}$t_5$, \tiny $\frac{1}{3}$} &  \\ \hline
        \multicolumn{3}{|c|}{\cellcolor[HTML]{C1C1C1}$t_4$, \tiny $\frac{1}{3}$}  &  \multicolumn{1}{c|}{$x_4$, \tiny $\frac{1}{3}$}  \\ \hline 
        \multicolumn{3}{|c|}{\cellcolor[HTML]{C1C1C1}$t_3$, \tiny $\frac{1}{3}$}  &  \multicolumn{1}{c|}{$x_3$, \tiny{$\frac{1}{3}$}}  \\ \hline 
        \multicolumn{3}{|c|}{\cellcolor[HTML]{C1C1C1}$t_2$, \tiny $\frac{1}{3}$}  & \multicolumn{1}{c|}{\cellcolor[HTML]{C1C1C1}$x_2$, \tiny $\frac{1}{3}$}  \\ \hline 
        \multicolumn{3}{|c|}{\cellcolor[HTML]{C1C1C1}$t_1$, \tiny $\frac{1}{3}$}  & \multicolumn{1}{c|}{\cellcolor[HTML]{C1C1C1}$x_1$, \tiny $\frac{1}{3}$}  \\ 
         \cline{1-3}  \cline{4-4}
        \multicolumn{4}{|c|}{\cellcolor[HTML]{C1C1C1}$c$, 1}  \\ \hline
        $v_1$ & $v_2$ & $v_3$ & $v_4$
    \end{tabular}
    \end{lrbox}
    \setlength{\fboxrule}{2pt}
    \ooalign{%
        \vspace{-58pt}
      \hss\usebox{\mytable}\hss \cr
      {\color{red}%
      \fbox{\phantom{\rule[-88pt]{\dimexpr\wd\mytable-56pt}{\dimexpr\baselineskip+5pt}}}
      }
      \vspace{15pt}
    }
    \caption{A laminar PB-instance with a laminar proportional bundle $W$ (indicated in grey), and a set of projects $T$ (indicated with a red border). The figure should be read in the same way as Table \ref{tab:ex_lamprop}.}
    \label{fig:counter_lp-core}
\end{figure}

The bundle $W=\{c, t_1, ..., t_6, x_1, x_2\}$ as indicated in grey in the figure is laminar proportional for limit $l=\frac{11}{3}$ (which is also its cost). However, $S = \{v_1, v_2, v_3\}$ is a blocking coalition. $S$ can afford $T$:  $|S|=3>\frac{\frac{8}{3}}{\frac{11}{3}}\cdot 4 = \frac{\text{cost}(T)}{l}\cdot n$, and for any voter $i\in S$, $u_i(T)=8> 7 = u_i(W)$. Therefore, $W$ is not in the core.

Note that all laminar PB-instances are approval-PB-instances.
In approval-PB-instances, the definitions of $(\alpha,T)$-cohesiveness and EJR simplify to the following \citep{peters2020proportionalPB}:
\begin{definition}[$T$-cohesiveness and EJR for approval-PB-instances] \label{def:EJR_PB_approval}
 A group of voters $S$ is $T$-cohesive for $T\subseteq C$ if $T$ is affordable with their share of the budget and they all approve all projects in $T$: $|S|\geq \frac{\text{cost}(T)}{l}\cdot n$ and $T\subseteq \cup_{i\in S}A_i$. A bundle $W$ satisfies \textit{approval-EJR} if for all $T$-cohesive groups $S\subseteq N$, there is a voter $i$ in $S$ who approves at least as many projects in $W$ as in $T$: $|W\cap A_i|\geq |T|$.
\end{definition} 
In the same way we can restrict our definition of PJR to the approval-based setting, and obtain the following:
\begin{definition}[Approval-PJR]\label{def:PJR_PB_approval}
 A bundle $W$ satisfies \textit{approval-PJR} if for all $T$-cohesive groups $S\subseteq N$, the number of projects in $W$ that is approved by at least one of the voters in $S$ is larger than the number of projects in $T$: $|W\cap \cup_{i\in S}A_i|\geq |T|$.
\end{definition}
 Using the definitions for EJR (Def. \ref{def:EJR_PB_approval}) and PJR (Def. \ref{def:PJR_PB_approval}) for approval-PB-instances, we see that in the counterexample in Figure \ref{fig:counter_lp-core}, the group $S$ is $T$-cohesive for given $T$, and there is no voter $i\in S$ such that $|W\cap A_i|\geq |T|$, neither is $|W\cap \cup_{i\in S}A_i|\geq |T|$.
\end{proof}
\paragraph{Theorem \ref{thm:LP-core_u-afford}}
\begin{quote}\it
    Laminar proportional bundles satisfy the core subject to u-afford.
\end{quote}
\begin{proof} We will prove this by induction over the structure of laminar profiles.\\
\textbf{Basis}: For unanimous profiles $P$ (Definition \ref{def:laminar_PB}, item 1), $W$ will consist of projects that are approved by every voter, so clearly for every group of voters $S \subseteq N$ and $T \subseteq C$ with $|S| \geq \frac{\text{cost}(T)}{l} \cdot n$, there is some voter in $S$ (namely all voters in $S$) who approves at least as many projects in $W$ as in $T$. This is true even in the general situation, without the restriction of u-afford.\\
\textbf{Inductive hypothesis}: Suppose that $W', W_1$, and $W_2$ are arbitrary bundles in respective laminar instances $(P', l'), (P_1, l_1)$, and $(P_2, l_2)$  with $C(P_1)\cap C(P_2)=\emptyset$ and $|P_1|\cdot l_2 = |P_2|\cdot l_1$, that  $W', W_1$, and $W_2$ are laminar proportional and are in the core subject to u-afford.\\
\textbf{Inductive step}: 
\begin{itemize}
    \item Suppose there is a unanimously approved project $c$ and $(P_{-c}, l-\text{cost}(c)) = (P', l')$ is laminar, $W'$ is a laminar proportional bundle for $(P', l')$, and $W=W'\cup \{c\}$ (Definition \ref{def:laminar_PB}, item 2). \\
    Assume for a contradiction that $W$ is not in the core. Then there must exist a group of voters $S$ and a set of projects $T$ such that $|S|\geq \frac{\text{cost}(T)\cdot n}{l}$, with $u_i(T)>u_i(W)$ for all voters $i\in S$. Because utilities are assumed to be additive and everyone approves project $c$, we know that for all voters $i\in S$, $u_i(T\backslash \{c\}) = u_i(T) - u_i(c)$ if $c\in T$ and $u_i(T\backslash \{c\}) = u_i(T)$ otherwise, and $u_i(W') = u_i(W) - u_i(c)$, so for all $i\in S$ \begin{equation}
    u_i(T\backslash\{c\})>u_i(W') \label{eq:uT>uW}
    \end{equation} 
    We distinguish two cases, with either $c\in T$ or $c\notin T$, and show that in both cases $W$ is in the core.
    \begin{enumerate}
        \item $c\in T$:\\
        Because we have chosen $S$ and $T$ such that $\frac{|S|}{n}\geq \frac{\text{cost}(T)}{l}$ and because cost$(T)\leq l$ (since by definition $|S|\leq n$), we find that 
    \begin{equation} \label{eq:costT-c}
        \frac{\text{cost}(T\backslash \{c\})}{l'}= \frac{\text{cost}(T) - \text{cost}(c)}{l - \text{cost}(c)} \leq \frac{\text{cost}(T)}{l}\leq \frac{|S|}{n}. 
    \end{equation}
    However, from the inductive hypothesis we know that $W'$ is in the core (subject to u-afford) in the election instance $(P',l')$, so for all $S'\subseteq N', T'\subseteq C'$ with $|S'|\geq \frac{\text{cost}(T')\cdot n}{l'}$, there is a voter $i'\in S'$ with $u_{i'}(W) \geq u_{i'}(T)$. In this instance, we can take $T'=T\backslash \{c\}$ and $S'=S$. As shown above, $S$ can afford $T'$ ($|S|\geq \frac{\text{cost}(T')\cdot n}{l'}$), so there is a voter $i'\in S'$ with $u_{i'}(W) \geq u_{i'}(T)$. This is a contradiction with \ref{eq:uT>uW}, which proves that $W$ is indeed in the core (subject to u-afford) if $c\in T$.
    \item $c\notin T$:\\
    In this case, equation \ref{eq:costT-c}    does not hold anymore, because not necessarily $\frac{\text{cost}(T)}{l - \text{cost}(c)} \leq \frac{\text{cost}(T)}{l}$,
    in fact the first fraction is greater because $c$ has a positive cost. Now suppose that $S$ can afford $T$ and that every voter in $S$ prefers $T$ to $W$. Then, since $c$ is unanimously preferred, all voters in $S$ prefer $T$ to $W\backslash\{c\}$, and even all voters in $S$ prefer $T\backslash\{t\}$, where $t$ is an arbitrary project in $T$, to $W\backslash\{c\}$,
    because we are in an approval voting setting (where the utility of a project a voter approves is 1 and the utility of all projects a voter does not approve is 0). However, according to our inductive hypothesis, if the voters in $S$ together could afford $T\backslash\{t\}$ in the situation where the budget is $l-\text{cost}(c)$, there would be a voter in $i\in S$ with $u_i(W\backslash\{c\})\geq u_i(T\backslash\{t\})$, since $W\backslash\{c\} =W'$ is in the core (subject to u-afford) there. Hence, $S$ cannot afford $T\backslash\{t\}$ in the instance $(P', l-\text{cost}(c))$. We now have that 
    \begin{equation} 
        \frac{\text{cost}(T)}{l}\leq \frac{|S|}{n} < \frac{\text{cost}(T) - \text{cost}(t)}{l - \text{cost}(c)}, 
    \end{equation}
    where $t$ was an arbitrary project in $T$, so all $t\in T$ must have a lower cost than $c$. Hence, under our restriction that there exists $t\in T$ with $\text{cost}(t)\geq\text{cost}(c)$, there is no $S$ that can block the winning bundle. Hence, if $c\notin T$, $W$ is in the core subject to u-afford.
    
    \end{enumerate}
    
    \item Suppose that $(P,l)$ is the sum of $(P_1, l_1)$ and $(P_2, l_2)$, i.e.that $P=P_1+P_2$ and $l=l_1+l_2$, and that $W = W_1\cup W_2$ (Definition \ref{def:laminar_PB}, item 3). Assume for a contradiction that $W$ is not in the core. Then there must exist a group of voters $S$ and a set of projects $T$ such that $|S|\geq \frac{\text{cost}(T)\cdot n}{l}$, with $u_i(T)>u_i(W)$ for all voters $i\in S$. Because $P_1$ and $P_2$ are strictly separated and voters can only approve projects from their own election instances, each voter only gets utility from the selected projects from his own instance, so we can divide $S$ into $S_1$ and $S_2$, and $T$ into $T_1$ and $T_2$ such that all voters from $S_1$ and projects from $T_1$ only occur in $P_1$ and all voters from $S_2$ and projects from $T_2$ only occur in $P_2$. Then we have that for all voters $i\in S_1$, $u_i(T_1)>u_i(W)$, and  for all voters $i\in S_2$, $u_i(T_2)>u_i(W)$.
    From $|S|\geq \frac{\text{cost}(T)\cdot n}{l}$ follows that
    \begin{eqnarray}
       |S_1|+|S_2| &\geq& \frac{(\text{cost}(T_1)+ \text{cost}(T_2))\cdot (n_1+n_2)}{l_1+l_2}. \label{eq:S1+S2}
    \end{eqnarray}
    Now, assume for a contradiction that both $|S_1| < \frac{\text{cost}(T_1)\cdot (n_1)}{l_1} $ and $|S_2| < \frac{\text{cost}(T_2)\cdot (n_2)}{l_2}$. Then from Equation \ref{eq:S1+S2} and the fact that $\frac{n_1}{l_1}=\frac{n_2}{l_2}$ (from the inductive hypothesis) we have that:
        \begin{eqnarray*}
        \resizebox{0.9\columnwidth}{!}{
        $
        \begin{aligned}
        \frac{\text{cost}(T_1)\cdot n_1}{l_1} + \frac{\text{cost}(T_2)\cdot n_2}{l_2} &>& |S_1| +|S_2|\\
       \text{cost}(T_1)\cdot \frac{n_1}{l_1} + \text{cost}(T_2)\cdot\frac{ n_2}{l_2}&\geq& \frac{(\text{cost}(T_1)+ \text{cost}(T_2))\cdot (n_1+n_2)}{l_1+l_2}\\
       (\text{cost}(T_1)+ \text{cost}(T_2))\cdot \frac{n_2}{l_2} &>& \frac{(\text{cost}(T_1)+ \text{cost}(T_2))\cdot (n_1+n_2)}{l_1+l_2}\\
       \frac{n_2}{l_2} &>& \frac{n_1+n_2}{l_1+l_2}\\
       \frac{n_2\cdot (l_1+l_2)}{l_2\cdot (l_1+l_2)} &>& \frac{(n_1+n_2)\cdot l_2}{(l_1+l_2)\cdot l_2}\\
       l_1\cdot n_2 + l_2 \cdot n_2 &>& l_2\cdot n_1 + l_2 \cdot n_2\\
       l_1\cdot n_2&>& l_2\cdot n_1 = l_2 \cdot \frac{n_2\cdot l_1}{l_2} = n_2\cdot l_1,  
        \end{aligned}
        $
       }
    \end{eqnarray*}
    which is clearly a contradiction. Hence, at least one of $|S_1| \geq \frac{\text{cost}(T_1)\cdot (n_1)}{l_1} $ and $|S_2| \geq \frac{\text{cost}(T_2)\cdot (n_2)}{l_2}$ must be true. Without loss of generality, assume that $|S_1| \geq \frac{\text{cost}(T_1)\cdot (n_1)}{l_1}$. Then, since $W_1$ is a core solution in the instance $(P_1, l_1)$, there exists $i \in S_1$ such that $u_i(W_1) \geq u_i(T_1)$. However, we already knew that for all voters $i\in S_1$, $u_i(T_1)>u_i(W)$. Since for voters $i$ from $S_1$ $u_i(W_1)=u_i(W)$, this is a contradiction, that shows that $W$ is in the core in the election instance $(P, l)$.
\end{itemize}
This completes the proof.
\end{proof}

\section{Additional Results}
\subsection{Relations Between Axioms}
 \subsubsection{Balanced Stable Priceability}\label{sec:BSP-LP}
 Note that the counterexample in Table \ref{tab:counterexample2} that shows that priceability does not imply LP (in Appendix \ref{sec:appendix-proofs}, \hyperlink{link:th:pr_not_imply_lp}{proof} of Theorem \ref{th:pr_not_imply_lp})) is not efficient. By electing $c_6$ instead of $c_3$ no agent's utility would decrease, but $v_3$'s utility would increase. Maybe making it efficient would make it laminar proportional, because then unanimous candidates have to be selected. Note however that electing $c_6$ instead of $c_4$ would make it efficient and still keep it priceable ($v_3$ could just spend her money on $c_6$ instead of on $c_4$), but still it is not laminar proportional because the instance without $c_6$ is not laminar proportional.\\
 However, if the payments would have been equally divided over the voters that approve a candidate, it would have been laminar proportional. In \cite{petersmarket}, a property that demands exactly this is defined: Balanced Stable Priceability (BSP). This property demands that a price system is balanced: voters that get utility from a candidate must all pay the same price for this candidate, and that the system is stable: there is no coalition of voters that wants to change their payments so that they get more utility (or pay less). As \cite{petersmarket} show, the bundles that satisfy BSP for a price $p$ are the same as the bundles selected by a variant of Rule X, and because Rule X returns laminar proportional bundles in laminar profiles, probably BSP implies LP. We can show easily by induction over laminar profiles that this is indeed the case.
 \begin{theorem}\label{thm:BSP-LP}
 Balanced Stable Priceability implies LP in MWV-instances.
 \end{theorem}
 \begin{proof}
  We will give an inductive proof to show this.  \\
 \textbf{Basis}: for unanimous profiles with $|C(P)|\geq k$, any candidate $c$ that is in $W$ gets at least some payment in the price system that supports $W$, and hence is in $A_i$ for some voter $i$, and because $P$ is unanimous, $c\in A_i$ for all voters $i$, so $W\subseteq C(P)$.\\
\textbf{Inductive Hypothesis}: Assume laminar profiles $(P',k')$, $(P_1, k_1)$ and $(P_2,k_2)$ are laminar and respective bundles $W'$, $W_1$, and $W_2$ are laminar proportional if they satisfy BSP, where $P'$ is not unanimous. \\
\textbf{Inductive Step}:
\begin{itemize}
    \item Suppose $c$ is a unanimously approved candidate, such that the instance $(P', k') = (P - \{c\}, k-1)$ and that $W$ satisfies BSP in the instance $(P, k)$. Then, by stability, $c\in W$. Assume for a contradiction that $c$ would not be selected, then all voters together would rather pay for $c$ and all give up one of the candidates they now pay for: then they would all get the same utility because they all approve $c$, and would have to pay less because they can divide the price for $c$ over them all. Hence $c$ is selected in $W$. Now the bundle $W$ without $c$, which we call to be the $W'$ from the inductive hypothesis, still satisfies BSP, because every voter pays the same amount for $c$ (because the price system is \textit{balanced}), and hence we can just subtract the price they all pay for $c$ from the total budget every voter gets. Then, by the inductive hypothesis, $W'$ is laminar proportional, so $W$ itself is laminar proportional as well.
    \item Suppose $(P,k)$ consists of two separate laminar MWV-instances $(P_1, k_1)$ and $(P_2, k_2)$. Define $W_1$ as the set of candidates in $W$ from $P_1$, and $W_2$ as the set of candidates in $W$ from $P_2$, so $W=W_1\cup W_2$. If $W$ satisfies BSP, then in the price system that witnesses this, voters from $P_1$ can only vote and pay for candidates in $P_1$, and voters from $P_2$ can only vote and pay for candidates in $P_2$, so we can split the price system to get a price system for both instances, which shows that both $W_1$ and $W_2$ satisfy BSP. According to the inductive hypothesis, then $W_1$ and $W_2$ are laminar proportional, so $W$ is laminar proportional. 
\end{itemize}
We have thus shown by induction over laminar profiles that if a winning bundle in a laminar profile satisfies BSP, it also satisfies LP.
\end{proof}

Note however, that this result relies on all projects having unit-cost. In the general case with non-unit costs, it does not hold anymore. To show this, we first have to give a definition of BSP for PB-instances.
Recall the definition of BSP for MWV-instances from \cite{petersmarket}. We still use approval votes (for compatibility with LP), but can have arbitrary costs. Requirement \textbf{E1} stays the same. Requirement \textbf{E5} changes a bit:\\
Condition for Stability: There exists no coalition of voters $S\subseteq N$, no bundle $(W', \{u'_i\}_{i\in N})$ ($W'\subseteq C\backslash W$) and no collections $\{p'_i\}_{i\in S}$ and $\{R_i\}_{i\in N}$ (with $R_i\subseteq W$ for each $i\in N$) such that all the following hold:
\begin{enumerate}
    \item For each $c\in W'$: there exists a value $\rho_c$ such that $p'_i(c)=u'_i(c)\cdot \rho_c$.
    \item For each $c\in W'$: $\sum_{i\in S}p'_i(c)>\text{cost}(c)$.
    \item For each $i\in S$: $p_i(W \backslash R_i)+p'_i(W')\leq \frac{\ell}{n}$.
    \item For each $i\in S$:  $(u_i(W\backslash R_i)+u'_i(W'), p_i(W\backslash R_i)+p'_i(W'))\succeq (u_i(W), p_i(W)).$
\end{enumerate}
Now we can answer the question whether BSP still implies LP in laminar profiles without the unit-cost assumption. In MWV-instances, part of the reason why BSP implies LP is because the `stability' part form BSP there requires that any unanimously approved candidate is selected in the winning bundle, since any voter would rather pay for a unanimously approved candidate than for a non-unanimously approved candidate they pay for now: a unanimous candidate costs less (costs are shared by all voters), and gives the same utility. In PB-instances, it will still give the same utility, but may not be affordable, or at least may not cost less. This is because it is possible that the unanimous candidate has a much higher cost than the other candidates. Hence, it is not the case that any bundle satisfying BSP in a laminar PB-instance is laminar proportional. 

As an example, take the profile in Table \ref{tab:example_BSP-LP-PB}, which is a laminar PB-instance with a budget of $\ell =10$. The grey bundle $W=\{c_1, c_2, c_3, c_4\}$ is Balanced Stable Priceable : every voter has an initial budget of $2\frac{1}{2}$, $v_1$ and $v_2$ can pay for $c_1$ and $c_2$, and $v_3$ and $v_4$ for $c_3$ and $c_4$. In a MWV instance, the bundle $\{c_2, c_4, c_5\}$ would be preferred, since all could share the cost for $c_5$, and still have a utility of 2. However, since $c_5$ has a higher cost here, it is not possible for any group of voters to afford $c_5$ and still have the same amount of utility. Hence, $W$ satisfies BSP, although it does not satisfy LP (since there is an non-selected unanimous candidate). One could argue about the fairness of the LP axiom here though: is there a clear reason why only selecting $c_5$  is more proportional than selecting $W$?
\begin{table}
    \centering
    \setlength{\tabcolsep}{15pt}
    \begin{tabular}{c c c c}
    \hline
        \multicolumn{2}{|c|}{\cellcolor[HTML]{C1C1C1}$c_1$, 3} &  \multicolumn{2}{|c|}{\cellcolor[HTML]{C1C1C1}$c_3$, 1} \\ \hline
        \multicolumn{2}{|c|}{\cellcolor[HTML]{C1C1C1}$c_2$, 2} &   \multicolumn{2}{|c|}{\cellcolor[HTML]{C1C1C1}$c_4$, 4}  \\ \hline
        \multicolumn{4}{|c|}{$c_5$, 10}  \\ \hline
        $v_1$ & $v_2$ & $v_3$ & $v_4$
    \end{tabular}
    \caption{Example of a PB-instance with a bundle satisfying BSP but not LP. It should be read in the same way as Table \ref{tab:ex_lamprop}.}
    \label{tab:example_BSP-LP-PB}
\end{table}

\subsubsection{Priceability does not imply EJR or the core in MWV-instances}
\begin{theorem}\label{thm:pr-not-EJR_MWV}
Priceability does not imply EJR in MWV-instances.
\end{theorem}
\begin{proof} Take an election instance $E$ with 3 voters $N=\{v_1, v_2, v_3\}$ and 6 candidates $C=\{c_1, ...c_6\}$, and let $k=3$. Let every voter approve four projects, namely one that only that voter approves and three that all three voters approve:  $A_i =\{c_i, c_4, c_5, c_6\}$. Suppose the winning committee of the election is $W = \{c_1, c_2, c_3\}$. For this committee, there is a price system with $p=1$ in which every voter $i$ pays the price of candidate $i$ and nothing else. The group $S=N$ is 3-cohesive, because $|S|=3=3\frac{n}{k}$, and all three voters agree on the projects $\{c_4, c_5, c_6\}$. However, there is no voter in $S$ who approves 3 or more projects in $W$, so $W$ does not satisfy EJR.  
\end{proof}

Since the core implies EJR, this example also shows that priceable committees are not necessarily in the core. 

\subsection{Properties of Rules}

\subsubsection{PAV does not satisfy PJR.}
As we mentioned in the main text, as shown in \cite[Figure 2]{peters2020proportionalPB_arxiv} that PAV does not give a proportional bundle without the unit cost assumption, and therefore PAV does not satisfy PJR. Here we give an exact argument why PAV does not satisfy approval-PJR (Definition \ref{def:PJR_PB_approval}).
\begin{proposition}
There exist approval-PB-instances where PAV does not satisfy PJR
\end{proposition}
\begin{proof}
 We use the example of Onetown from \cite{peters2020proportionalPB_arxiv}. The group of voters in Leftside are $T$-cohesive  for $T=\{L_1,L_2,L_3\}$: they can with their share of the money afford all projects in $T$ and do all approve all projects in $T$. However, the amount of projects in the bundle $W$ that PAV returns that at least one of the voters in Leftside ($S$) approves of is $|W\cap_{i\in S}A_i| =2$, which is less than the number of projects in $T$. 
\end{proof}

\subsubsection{Phragm\'{e}n satisfies the core and EJR in laminar MWV-instances.}
From Theorem \ref{thm:LP-core_u-afford} follows the following:
Since Phragm\'{e}n satisfies LP in MWV-instances, in laminar MWV-instances the winning bundle of Phragm\'{e}n will also satisfy the core and EJR.
\begin{corollary}\label{cor:Phragmen-core,ejr-laminar_instances}
In laminar MWV-instances, Phragm\'{e}n satisfies the core and EJR.
\end{corollary}

\end{document}